\documentclass[11pt]{article}
\pdfoutput=1
\usepackage[margin=1in]{geometry}

\usepackage[hyphens]{url}
\usepackage[
  bookmarks=true,
  bookmarksnumbered=true,
  bookmarksopen=true,
  pdfborder={0 0 0},
  breaklinks=true,
  colorlinks=true,
  allcolors=teal
]{hyperref}

\usepackage[utf8]{inputenc}
\usepackage{titlesec}
\usepackage{comment}
\usepackage{amsmath}
\usepackage{amsfonts}
\usepackage{amsthm}
\usepackage{subcaption}
\usepackage[round]{natbib}

\newtheorem{theorem}{Theorem}[section]
\newtheorem{corollary}[theorem]{Corollary}

\newtheorem{lemma}[theorem]{Lemma}
\newtheorem{proposition}[theorem]{Proposition}

\theoremstyle{definition}
\newtheorem{definition}[theorem]{Definition}

\usepackage{booktabs} 
\usepackage[ruled,linesnumbered]{algorithm2e} 

\SetAlFnt{\small}
\SetAlCapFnt{\small}
\SetAlCapNameFnt{\small}
\SetAlCapHSkip{0pt}
\IncMargin{-\parindent}
\usepackage[capitalize]{cleveref} 
\usepackage[most]{tcolorbox}
\tcbuselibrary{listings}
\usepackage{placeins}
\usepackage{subcaption}
\usepackage{tabularx}
\usepackage{threeparttable}
\usepackage{aliascnt}
\usepackage{mathtools}
\usepackage{soul}
\setstcolor{red}
\newtcblisting{promptbox}[1]{
    colback=gray!5,
    colframe=gray!50,
    fonttitle=\bfseries\sffamily,
    title=#1,
    breakable,
    listing only, 
    listing options={
        style=tcblatex,
        basicstyle=\small\ttfamily,
        breaklines=true,
        columns=fullflexible,
        literate={ \{ }{{{\{}}}1 { \} }{{{\}}}}1 
    },
    before skip=10pt plus 2pt,
    after skip=10pt plus 2pt,
    arc=2pt
}

\newcommand{\PVC}{\mathit{PVC}}

\crefname{theorem}{Theorem}{Theorems}
\crefname{definition}{Definition}{Definitions}
\crefname{corollary}{Corollary}{Corollaries}

\newaliascnt{condition}{theorem}
\newtheorem{condition}[condition]{Condition}
\aliascntresetthe{condition}
\crefname{condition}{Condition}{Conditions}
\Crefname{condition}{Condition}{Conditions}

\allowdisplaybreaks

\usepackage{titlesec}

\titleformat{\paragraph}[runin]
  {\normalfont\normalsize\bfseries} 
  {\theparagraph}                   
  {0.5em}                           
  {}                                
  [.]                               

\titlespacing*{\paragraph}
  {0pt}                             
  {3.25ex plus 1ex minus 0.2ex}     
  {0.33em}                          
  
\title{Finding Common Ground in a Sea of Alternatives}

\author{
    Jay Chooi$^1$,
    Paul G\"olz$^2$,
    Ariel D. Procaccia$^1$,
	Benjamin Schiffer$^3$, and
	Shirley Zhang$^1$ \\[.5em]
    \normalsize$^1$ Paulson School of Engineering and Applied Sciences, Harvard University \\
    \normalsize$^2$ School of Operations Research and Information Engineering, Cornell University \\
    \normalsize$^3$ Department of Statistics, Harvard University \\[.3em]
    \normalsize \texttt{jeqin\_chooi@college.harvard.edu}, \texttt{paulgoelz@cornell.edu}, \texttt{arielpro@seas.harvard.edu}, \\ \normalsize  \texttt{bschiffer1@g.harvard.edu}, \texttt{szhang2@g.harvard.edu}
}
\date{}
\begin{document}

\maketitle

\begin{abstract}

We study the problem of selecting a statement that finds common ground across diverse population preferences. Generative AI is uniquely suited for this task because it can access a practically infinite set of statements, but AI systems like the Habermas machine~\citep{TBJ+24} leave the choice of generated statement to a voting rule. What it means for this rule to find common ground, however, is not well-defined.

In this work, we propose a formal model for finding common ground in the infinite alternative setting based on the \emph{proportional veto core} from social choice. To provide guarantees relative to these infinitely many alternatives and a large population, we wish to satisfy a notion of proportional veto core using only query access to the unknown distribution of alternatives and voters.

We design an efficient sampling-based algorithm that returns an alternative in the (approximate) proportional veto core with high probability and prove matching lower bounds, which show that no algorithm can do the same using fewer queries. On a synthetic dataset of preferences over text, we confirm the effectiveness of our sampling-based algorithm and compare other social choice methods as well as LLM-based methods in terms of how reliably they produce statements in the proportional veto core.

\end{abstract}

\section{Introduction}
In many countries, polarization has reached unhealthy levels~\citep{Silver22},
dividing politics into “us” versus “them”.
This division exacts a heavy toll: it reduces trust between people, entails political gridlock and democratic erosion, and increases political violence~\citep{Lee22,Piazza23,MRS18,Orhan22}.
To stem the tide of polarization, recent initiatives encourage people to discuss politics across divides and find points of agreement.
The nonprofit \emph{Braver Angels}\footnote{\url{https://braverangels.org/}}, for example, brings groups of Democrats and Republicans face to face, and \emph{citizens' assemblies}~\citep{DBC+19,OECD20} convene a representative sample of constituents to design common policy proposals.
Major obstacles, however, stand in the way of scaling these efforts to a society-wide level, including the cost of trained facilitators and cognitive limits that prevent high-quality deliberation in large groups~\citep{Fishkin09}.
The recent success of \emph{large language models (LLMs)} in processing natural language text shows promise in overcoming these obstacles~\citep{Landemore24,MC25,SVD+23,dembrane2023democratic,KSI+23}.

Perhaps the most prominent AI tool with this goal is the \emph{Habermas machine}, presented by \citet{TBJ+24} in \emph{Science}.
To help a group deliberate on a topic, the Habermas machine transforms opinion statements from the group members into possible group statements, which are pieces of text representing the common ground of opinions.
\citet{TBJ+24} test the Habermas machine in the context of a citizens’ assembly and find that groups interacting with the Habermas machine converge in their beliefs towards a common position, decreasing polarization.

Besides LLM components (for the generation of text and prediction of agreement), the Habermas machine also crucially relies on \emph{voting} to select among generated proposals.
Voting happens on two levels:
first, to generate alternative statements, 16 statements are sampled from a fine-tuned LLM, ranked by a proxy model for each group member, and only the winner according to the Schulze voting rule~\citep{Schulze11} is kept;
second, four alternative statements generated in this fashion are ranked by the actual group members, and the winning statement is again selected by Schulze.

For finding common ground, Schulze may not be the right tool for selecting among statements because this voting rule is inherently majoritarian: it readily overrides the preferences of a minority to satisfy even a slight majority, as evidenced by axioms such as Condorcet consistency and the mutual majority criterion.
While \citet{TBJ+24} 
find that minority opinions have an influence on the winning statement,\footnote{Specifically, \citet{TBJ+24} embed the opinions by all group members using a text embedding model, and run a regression to predict the embedding of the winning statement as a mixture of these opinion statements. Group members are partitioned into a ``majority'' and ``minority'' group based on whether they answered ``somewhat agree/agree/strongly agree'' or ``somewhat disagree/disagree/strongly disagree'' in an initial question on the topic. \citet{TBJ+24} find that minority agents tend to have non-negligible weight in this regression.}
this is an empirical finding and only holds on average across their experiments.
To robustly find common ground, we should aim for guarantees on the influence of minorities that hold for every possible group deliberation.

What, then, can social choice offer as an axiom for finding common ground and as a voting rule to be embedded inside the Habermas machine? 
We argue that the \emph{proportional veto core (PVC)} of \citet{Moulin81,Moulin82} is the right notion since it gives explicit outcome guarantees to cohesive voter coalitions of all sizes.
Consider a coalition of, say, 30\% of the participants and some alternative statement $c$.
If the coalition can find a set with $100\% - 30\%=70\%$ of the alternatives that they all rank above $c$, the PVC guarantees that $c$ will not be selected.
In this way, minority groups are guaranteed an influence on the outcome, and they can avoid the statements they disagree with the most, which is important for a claim of common ground.
To pursue this guarantee, the Habermas machine could use a voting rule whose output always lies in the PVC (like Vote by Veto \citep{Moulin81}), one could constrain any voting rule to alternatives in the PVC, or one could evaluate voting rules (classical or LLM-based) in terms of how reliably their output lies in the PVC.

One more challenge of using social choice in settings like the Habermas machine is that the space of alternatives is tremendously large and often infinite. This leads to the natural question:
Are the 16 statements sampled i.i.d.\ from the fine-tuned LLM enough to find common ground, and are the four statements resulting from independent runs of sampling and simulated voting enough?
How can we capture what ``enough'' is?
In this paper,
\begin{quote}
we define proportional veto core guarantees with respect to a very large set of statements and design practical algorithms that satisfy them.
\end{quote}
Our pursuit of guarantees with respect to a large set of textual alternatives is closely aligned with the \emph{generative social choice} framework by \citet{FGP+26}.
Like them, we can only access this space of alternatives indirectly through queries, which in our case means sampling from an unknown distribution over alternatives.

\subsection{Our Contributions}

We start by giving a formal mathematical model of what it means for a statement to find common ground when there are an infinite number of alternatives.
We propose using the \emph{$\epsilon$-proportional veto core ($\epsilon$-PVC)}, which extends the standard PVC to the setting with a distribution over alternatives. We define the \emph{critical $\epsilon$} of an alternative as the smallest $\epsilon$ for which that alternative is in the $\epsilon$-PVC. Therefore, a smaller critical $\epsilon$ indicates that a statement better represents common ground. 

To model information elicitation when there are infinite alternatives, we study a query-based information model with two types of queries. First, we explore \emph{generative} queries, which return alternatives sampled from an unknown distribution of interest over the alternative space, such as the output of an LLM with respect to a given prompt. Second, we explore \emph{discriminative} queries, which return user preferences over alternatives. We focus on two types of discriminative queries: \textit{min-queries} that return a voter's least favorite alternative among a given set and \textit{pairwise queries} that return a voter's preference between two given alternatives.

\subsubsection*{Theoretical Results}
Our theoretical work focuses on providing upper and lower bounds for how many queries are necessary to find an element in the $\epsilon$-PVC. We begin by showing that the $\epsilon$-PVC always has size at least $\epsilon$. We then present a query-based algorithm that with high probability finds an element of the $\epsilon$-PVC using $\tilde{O}(1/\epsilon^2)$ generative queries and min-queries. We prove that no algorithm can always find an element of the $\epsilon$-PVC using fewer than $\Omega(1/\epsilon^2)$ generative queries and min-queries, which shows that our algorithm is worst-case optimal for both types of queries.

Pairwise queries are perhaps the easiest information elicitation method from the voters' perspective. Therefore, we show that our algorithm can be adapted to use pairwise  queries while giving the same theoretical guarantees. The key modification is showing that with $O(nm)$ pairwise queries, we can find an element of the PVC for finite sets of $n$ voters and $m$ alternatives. We provide an efficient method for doing this and further show that there is a matching lower bound for this step; i.e., no algorithm can find an element of the PVC with fewer than $\Omega(nm)$ pairwise queries for $n$ voters and $m$ alternatives. 

\subsubsection*{Experimental Results}
We conclude by empirically evaluating the critical $\epsilon$ achieved by different methods for selecting alternatives on simulated preference data. We obtain diverse preference data sets for a variety of topics by using ChatGPT to generate both alternatives and preference data for different personas. To model a query-based setting, the different voting rules are only given access to a small subset of voters and alternatives and must select an alternative from this subset.

Our first empirical result confirms that the sampling-based approach from the theory section finds alternatives in the $\epsilon$-PVC for small $\epsilon$. We then show that different social choice methods such as Schulze and plurality can select alternatives with large critical $\epsilon$, indicating that these methods may not necessarily choose a good common-ground statement. Finally, we conclude our experiments by studying generative voting rules that use LLMs to generate common-ground statements. We show that these generative voting rules do generate statements with small critical $\epsilon$, which indicates that the notion of PVC aligns well with how an LLM tries to find common ground. We also show that when the LLM is not properly tuned to the population, the generative voting rules give statements with significantly higher critical $\epsilon$ compared to our algorithm.

\subsection{Related Work}
Many recent works have studied how to use AI for finding consensus on a large scale~\citep{KLNP+19, gudino2024large,Landemore24,KSI+23,TBJ+24,BCST+22,BL26}.
The most similar of these methods is the Habermas machine~\citep{TBJ+24}, which we already discussed at length. 
At the end of \cref{sec:model}, we further discuss the clone-proofness axiom satisfied by the Schulze rule, and why our approach to common ground does not require it.

Finding common-ground statements is an important goal of online deliberation platforms like \emph{Polis}~\citep{SBES+21} and \emph{Remesh}\footnote{\url{https://www.remesh.ai/}}, on which users submit statements and rate their agreement with the statements written by others.
Both platforms~\citep{compdem_group_informed_consensus,KSI+23} present user-submitted statements based on a ``bridging-based ranking''~\citep{OT23}, which
measures if a statement has high approval rates in all clusters of users.
Our approach is much more general because we give guarantees to all subsets of voters rather than only some prespecified groups.

The idea of finding common-ground and bridging has also been popular in the context of annotating social media posts \citep{de2025supernotes,wojcik2022birdwatch}. For example, \citet{wojcik2022birdwatch} study a bridging-based algorithm for finding common-ground and selecting notes to append to tweets. Furthermore, \citet{slaughter2025community} show that this approach to finding common ground successfully reduces engagement with and proliferation of misinformation on social media platforms. While these methods for finding common ground have had practical success, they lack formal theoretical guarantees.

The PVC was first studied by \citet{Moulin81,Moulin82} in the traditional social choice setting of perfect information about a finite number of voters and alternatives. In the same setting, \citet{ianovski_computing_2023} give a polynomial time algorithm for computing the PVC and a neutral and anonymous algorithm for finding an element from the PVC.
Ideas from these algorithms inspire some of our algorithmic procedures, as discussed in later sections.
More recently, the PVC has been studied in a variety of contexts including approval ballots \citep{halpern2025proportional}, federated learning \citep{chaudhury2024fair}, and policy aggregation \citep{a2024policy}. The settings studied by \citet{chaudhury2024fair} and \citep{a2024policy} both can have an infinite number of alternatives, which lead these works to notions similar to the $\epsilon$-PVC.

Like us, 
the work on \emph{generative social choice}~\citep{FGP+26,boehmer2025generative} uses generative models for social choice over an infinite space of textual alternatives.
Similar to our setting, they study generative queries that generate alternatives and discriminative queries that access voter preferences.
A major difference, however, is that their goal is to find a set of \emph{several} representative statements, which enables them to give guarantees to minority groups without finding common ground within each statement.
A second major difference is that their space of alternatives lacks our structure of a probability distribution.
As a result, their generative queries (finding a statement maximizing a function of voter preferences) are much more complex and much harder to implement than our query of sampling from a fixed distribution.

Beyond the generative social choice model, there has also been significant work on applying social choice techniques to settings with many or infinite alternatives and limited information in the setting of AI alignment \citep{ge2024axioms,procaccia2025clone,golz2025distortion,conitzer2024social} and AI evaluation \citep{procaccia2025metritocracy,zhang2024}. Due to the size of the alternative space in these settings, these works also focus on choosing subsets of alternatives and voters to give approximate or probabilistic generalization guarantees.

\section{Model}\label{sec:model}

Let $\mathcal{M}$ be a (possibly infinite) set of alternatives and let $\mathcal{N}$ be a finite (but potentially very large) set of voters. Each voter $i \in \mathcal{N}$ has a ranking over all alternatives in $\mathcal{M}$. For alternatives $a,b \in \mathcal{M}$ and voter $i \in \mathcal{N}$, we use $a \succ_i b$ to mean that voter $i$ prefers alternative $a$ to alternative $b$. Importantly, we do not assume perfect information about the preferences of the voters. 

Instead, we assume that we can access information about the alternatives and voters using two types of queries: \textbf{generative queries} and \textbf{discriminative queries}. Generative queries allow us to access the set of alternatives $\mathcal{M}$. Specifically, let $\mathcal{D}$ be a fixed distribution over $\mathcal{M}$. For a set of alternatives $M \subseteq\mathcal{M}$, define $\mu_{\mathcal{D}}(M)$ as the measure of the set $M$ for distribution $\mathcal{D}$. A generative query returns a random sample from the distribution $\mathcal{D}$. Note that no information about the distribution $\mathcal{D}$ or the rankings of the voters is known a priori.

Discriminative queries allow us to gather information about voter preferences. We study two types of discriminative queries. The first type of discriminative query we study is a \emph{min-query}, which takes as input a voter $i \in \mathcal{N}$ and a set of alternatives $X \subseteq \mathcal{M}$ and returns voter $i$'s least favorite alternative in $X$. The second type of discriminative query we study is a \emph{pairwise query}, which takes as input a voter $i \in \mathcal{N}$ and two alternatives $a,b \in \mathcal{M}$, and returns which alternative in $\{a,b\}$ is ranked higher by voter $i$. Our main goal is to use a combination of generative and discriminative queries to efficiently find an alternative in $\mathcal{M}$ that finds common ground for the entire population $\mathcal{N}$.

We argue that the $\epsilon$-Proportional Veto Core ($\epsilon$-PVC) is a mathematically principled way of finding common ground. Informally, an alternative $a$ is blocked by a coalition of voters if all voters in the coalition prefer a sufficiently large set of alternatives to $a$. The $\epsilon$-PVC is simply the set of alternatives that is not blocked by any coalition of voters.
By ruling out alternatives for which 
some sizable minority ranks a large block of options higher, the $\epsilon$-PVC excludes polarizing alternatives.

\begin{definition}\label{def:epsilon_pvc}
An alternative $a \in \mathcal{M}$ is $\epsilon$-\textbf{blocked} by a coalition $T \subseteq\mathcal{N}$ of voters if there exists a blocking set of alternatives $S \subseteq\mathcal{M}$ such that
\begin{itemize}
    \item Each voter $i \in T$ prefers each alternative $b \in S$ to $a$ (i.e. $b \succ_i a$), and
    \item $\mu_{\mathcal{D}}(S) > 1 - \frac{|T|}{n} + \epsilon.$
\end{itemize}
The \textbf{$\epsilon$-Proportional Veto Core ($\epsilon$-PVC)} is the set of all alternatives that are not $\epsilon$-blocked by any coalition of voters.
\end{definition}

When $\mathcal{N}$ and $\mathcal{M}$ are finite, the Proportional Veto Core (PVC)~\citep{Moulin81} refers to the special case when $\epsilon = 0$ and $\mathcal{D}$ is the uniform distribution over all alternatives in $\mathcal{M}$. See Appendix \ref{app:pvc} for more details.

A higher $\epsilon > 0$ makes it harder for coalitions to block alternatives and therefore makes the $\epsilon$-PVC a less demanding notion.
Hence, the $\epsilon$-PVC naturally leads to a metric for how well an alternative $a \in \mathcal{M}$ does at finding common ground for the voters $\mathcal{N}$.
Intuitively, the critical $\epsilon$ for an alternative is the smallest fraction of voters that must be ignored in each coalition so that no coalition will veto that alternative. Therefore, a critical $\epsilon$ closer to $0$ implies that an alternative is better for finding common ground while a larger critical $\epsilon$ implies that an alternative is worse at it.

\begin{definition}\label{def:crit_eps}
    For any alternative $a \in \mathcal{M}$, the \textbf{critical $\epsilon$ }for $a$ is the smallest $\epsilon \ge 0$ such that $a$ is in the $\epsilon$-PVC.
\end{definition}

\subsubsection*{Role of the Statement Distribution}
$\mathcal{D}$ plays a crucial role in our notion of common ground since it determines how many alternatives a coalition $T$ must agree on to block an alternative $a$.
For example, let $\mathcal{D}$ be defined by an LLM, which produces a neutral summary with 99\% probability, but sometimes insults a minority instead.
The $\epsilon$-PVC allows the small minority to block the insulting statements, whose small measure in the distribution reflects its fringe nature.
This marks a major difference in perspective to \citet{TBJ+24}, who use the Schulze voting rule due to its clone-proofness.
If the Habermas machine samples 99 neutral statements and one insulting statement from the LLM and voters see the neutral statements as interchangeable, clone-proofness actually ensures that a slight majority preference for the insulting statement makes it the winner. In addition, clone-proofness is a brittle axiom that only applies to exact clones, so similar but not identically ranked alternatives can still have a large effect on the outcome. \footnote{While the winner of Schulze on samples from $\mathcal{D}$ is independent of the multiplicities of identical samples, we want to caution that the precondition of clone-proofness, requiring alternatives to be treated the same by \emph{every} voter, makes it a brittle axiom. For example, \citet{PS24} give an election where Schulze does not satisfy clone-proofness for alternatives that are almost clones. In this way, multiplicities can still have unpredictable effects under Schulze.}

As illustrated in the above example, having a distribution over alternatives enables us to guarantee rights to coalitions by forming a yardstick for how much influence each of them deserves.
This does mean, of course, that the distribution must be carefully chosen to be broadly acceptable to the voters.
In many cases, the statement distribution generated by an LLM might be a sufficiently neutral ground\,---\,say, from a base model, which seem to be more politically moderate than post-trained models~\citep{Rozado24}, or a model specifically fine-tuned for this purpose.
Another option could be to define $\mathcal{D}$ in terms of the distribution of voters as we do in our experiments: we draw characteristics of a random voter and ask the LLM to answer from their perspective.
We continue our discussion on the choice of distribution $\mathcal{D}$ in \cref{sec:distr}.

\section{Basic Properties}
To begin our investigation of the $\epsilon$-PVC, we start by bounding its size and developing a polynomial-time algorithm that can compute the critical $\epsilon$ and the $\epsilon$-PVC if $\mathcal{D}$ is the uniform distribution over a given set of alternatives.

\subsection{Size of $\epsilon$-PVC}
As a warm-up, we study the size of the $\epsilon$-PVC, which will tell us how many generative queries are necessary to even see an element of the $\epsilon$-PVC.
We begin by showing that, for any $\epsilon>0$, the $\epsilon$-PVC is not only non-empty but always has a size (with respect to the measure $\mu_{\mathcal{D}}$) of at least $\epsilon$. Previous papers that have looked at variants of the $\epsilon$-PVC~\citep{a2024policy, chaudhury2024fair} give similar proofs of existence, but do not explicitly study the size of the $\epsilon$-PVC.

\begin{proposition}\label{prop:epsilon_size}
    For any instance of the problem and any $\epsilon$, the size of the $\epsilon$-PVC (with respect to the measure $\mu_{\mathcal{D}}$) is at least $\epsilon$.
\end{proposition}
\begin{proof}[Proof sketch]
    The key idea of this proof relies on an algorithm we call vote by $\gamma$-veto (Algorithm \ref{alg:veto-by-consumption-text} in \cref{sec:epsilon_size_proof}), which is guaranteed to return a set of alternatives $C$ with mass at least $\epsilon$ such that all alternatives in $C$ are also in the $\epsilon$-PVC. In this algorithm, the voters are ordered randomly, and then in this order each voter vetoes their least favorite $\frac{1-\epsilon}{n}$ mass of alternatives that has not yet been vetoed. Clearly, once all $n$ voters have vetoed $\frac{1-\epsilon}{n}$ mass, the remaining alternatives must have mass at least $\epsilon$. The proof that the remaining alternatives are in the $\epsilon$-PVC follows similarly to the proof for the Vote by Veto rule of \citet{Moulin82}. Intuitively, any alternative $a$ not in the $\epsilon$-PVC must get removed before the algorithm ends, as the last voter to veto from the blocking set for this alternative $a$ is guaranteed to veto $a$ if $a$ has not yet been vetoed.
    For the full proof, see \cref{sec:epsilon_size_proof}.
\end{proof}

A direct corollary of \cref{prop:epsilon_size} is that, among sufficiently many generative queries, there will with high probability be an alternative in the $\epsilon$-PVC:
\begin{corollary}\label{cor:suff_many}
    For any $\epsilon, \delta > 0$, a set of $m = \frac{\log(1/\delta)}{\epsilon}$ generative queries will with probability $1-\delta$ contain an element of the $\epsilon$-PVC.
\end{corollary}
\begin{proof}
    Let $M$ be the result of $m = \frac{\log(1/\delta)}{\epsilon}$ generative queries and let $C_\epsilon$ be the $\epsilon$-PVC. Then
    \[
        \Pr(M \cap C_{\epsilon} = \emptyset) = (1- \mu_{\mathcal{D}}(C_\epsilon))^m \le (1-\epsilon)^{m} \le e^{-\epsilon m} = \delta,
    \]
    where the first inequality is by Proposition \ref{prop:epsilon_size}. Therefore, the probability that at least one element of $M$ is in $C_{\epsilon}$ is at least $1-\delta$.
\end{proof}

While $O(1/\epsilon)$ generative queries contain at least one element of the $\epsilon$-PVC with high probability, we show in Section \ref{sec:lower_bounds} that a larger number, specifically $\Omega(1/\epsilon^2)$, of generative queries is needed to \emph{identify} an element of the $\epsilon$-PVC. Therefore, identifying an element of the $\epsilon$-PVC with generative queries is statistically harder than simply generating an element of the $\epsilon$-PVC.

In contrast to the traditional PVC, the $0$-PVC can, in general, be empty:
\begin{proposition}
For an infinite number of alternatives, the $0$-PVC can be empty.
\end{proposition}
\begin{proof}
Let $\mathcal{M} \coloneqq \mathbb{N}_{\geq 1}$, and let each alternative $i \geq 1$ have mass $\frac{1}{2^i}$.
Suppose that all voters prefer alternatives with larger numbers.
We prove, by contradiction, that no alternative $i$ is in the $0$-PVC.
Indeed, any such $i$ is blocked by the grand coalition $T = \mathcal{N}$ and the alternative set $S =\{i+1\}$, since all voters prefer $i+1$ over $i$ and since $\mu_{\mathcal{D}}(S) = \frac{1}{2^{i+1}} > 0 = 1 - \frac{|T|}{n}$.
\end{proof}

\subsection{Computing Critical $\epsilon$}\label{subsec:compute-critical-epsilon}
As mentioned in the previous section, the critical $\epsilon$ of an alternative measures how well a given alternative does at finding common ground.
In this section, we give an algorithm for finding the critical $\epsilon$ of a given alternative in polynomial time when the number of alternatives is finite and $\mathcal{D}$ is the uniform distribution over $\mathcal{M}$.
Our algorithm generalizes the algorithm of \citet{ianovski_computing_2023} for deciding if an alternative is in the traditional PVC.
By running this algorithm for all alternatives, we can construct
the entire $\epsilon$-PVC in polynomial time in the same setting.

\begin{algorithm}[ht]
\caption{Computing the critical $\epsilon$}
\label{alg:calculate_critical}
\DontPrintSemicolon

\KwIn{Alternative $a$, $|\mathcal{N}| = n$, $|\mathcal{M}| = m$}

    Construct a flow graph $G$ as follows: \;
    \Indp
    Connect a source node to $n$ voter nodes, each edge with weight $m$\;
    Connect $m-1$ nodes (for all alternatives except $a$) to a sink node, each with edge weight $n$\;
    Connect voter node $i$ to alternative node $j$ with infinite weight if voter $i$ prefers $a$ to alternative $j$\;
    \Indm

    Let $K$ be the size of the min-cut for graph $G$ \;

    \Return $\frac{2nm - n - K}{nm} - 1$
\end{algorithm}

\begin{theorem}
    When $\mathcal{M}$ is finite, Algorithm \ref{alg:calculate_critical} computes the critical $\epsilon$ of a given alternative in polynomial time.
\end{theorem}
\begin{proof}[Proof Sketch]
Our proof generalizes the ideas of \citet{ianovski_computing_2023} for computing the PVC.
Their algorithm checks whether an alternative $a$ is in the PVC by constructing a graph $G_a$ such that there exists a complete biclique in $G_a$ if and only if $a$ is not in the PVC. We show that for the same graph $G_a$, the size of the maximum biclique is larger than $(1+\epsilon)nm$ if and only if that alternative is not in the $\epsilon$-PVC. In order to efficiently compute the size of the maximum biclique, we follow the same method as \citet{ianovski_computing_2023} of constructing the flow graph in \Cref{alg:calculate_critical} such that the weight of the max flow $K$ in this graph plus the size of the maximum biclique in $G_a$ equals $2nm - n $. In the last line of \Cref{alg:calculate_critical}, we use this fact to find the smallest $\epsilon$ for which the maximum biclique of $G_a$ is larger than $(1+\epsilon)nm$. The runtime of the algorithm is polynomial because the min-cut can be found in polynomial time (\Cref{subsubsec:critical-epsilon-runtime}). For the detailed proof, see \Cref{sec:deferred-proofs-for-eps-pvc-critical-epsilon}.
\end{proof}

\section{Algorithmic Results}\label{sec:algorithm}

In this section, we present our main algorithmic results for efficiently finding an element of the $\epsilon$-PVC. We give a simple algorithm using generative queries and min-queries that finds an alternative in the $\epsilon$-PVC, with high probability, using $\tilde{O}(1/\epsilon^2)$ queries of each type. A key takeaway from this result is that the number of generative and discriminative queries does \emph{not} depend on the size of $\mathcal{N}$ (the number of voters) or the distribution over alternatives. This is important because the population may be large and the generative distribution altogether unknown to the algorithm. Then, we discuss how the same algorithm can be modified to use pairwise discriminative queries instead of min-queries.

\subsection{Min-Queries}
We now present and analyze Algorithm \ref{alg:find_epsilon_PVC_element}, which finds an element of the $\epsilon$-PVC using generative queries and min-queries.

\begin{algorithm}[ht]
\caption{Finding an element of the $\epsilon$-PVC}
\label{alg:find_epsilon_PVC_element}
\DontPrintSemicolon

\textbf{Input:} $\mathcal{N}, \mathcal{M}, \delta, \epsilon$ \;

\textbf{Output:} An alternative that is with probability $1-\delta$ an element of the $\epsilon$-PVC\;

\vspace{0.5em}
    Set $\tau = \frac{8}{\epsilon^2} \, \log\left(\frac{32}{\epsilon^2\delta}\log\left(\frac{32}{\epsilon^2\delta}\right)\right)$. \label{line:tau}

    Let $M$ be a set of alternatives that is the result of $\tau$ generative queries. \label{line:choose_M}

    Initialize $X = M$.

    \While{$|X| > 1$}{ \label{line:loop_start}
        Select a random voter $i \in \mathcal{N}$.
        
        Do one min-query to find voter $i$'s least favorite alternative among $X$. \label{line:discrim}
        
        Remove that alternative from $X$. \label{line:loop_end}
    }
    \Return{the remaining element in $X$}
\end{algorithm}

\begin{theorem}\label{thm:alg_result}
    Using ${O}\left(\epsilon^{-2} \, \log\left(\frac{1}{\epsilon \delta}\right)\right)$ generative and min-queries, Algorithm \ref{alg:find_epsilon_PVC_element} returns an element of the $\epsilon$-PVC with probability at least $1-\delta$. 
\end{theorem}

\begin{proof}
    By construction, Algorithm \ref{alg:find_epsilon_PVC_element} uses $\tau$ generative queries and $\tau-1$ min-queries, which gives the desired bound on the number of queries, observing that, for small enough $\epsilon, \delta$,
    \[ \textstyle \tau = \frac{8}{\epsilon^{2}} \, \log\left(\frac{32}{\epsilon^2\delta}\log\left(\frac{32}{\epsilon^2\delta}\right)\right) \leq \frac{8}{\epsilon^{2}} \, \log \left( \frac{32^2}{\epsilon^4 \delta^2}\right) \leq \frac{8}{\epsilon^{2}} \, \log \left( \frac{32^4}{\epsilon^4 \delta^4}\right) = \frac{32}{\epsilon^{2}} \, \log \frac{32}{\epsilon \, \delta} = O\left(\epsilon^{-2} \, \log \left(\frac{1}{\epsilon \, \delta}\right)\right).\]

    Let $N$ be the set of voters that are selected in Lines \ref{line:loop_start}--\ref{line:loop_end} of Algorithm \ref{alg:find_epsilon_PVC_element}. Note that a voter $i \in \mathcal{N}$ could be chosen in two different iterations of the while loop, in which case we view these as distinct copies of the voter in the set $N$. This means we must have that $|N| = \tau -1$. 
    
    We first show that any alternative in the PVC for voters $N$ and alternatives $M$ will, with probability $1-\delta$, be in the $\epsilon$-PVC for $\mathcal{N}, \mathcal{D}$. We defer the proof of this result to Lemma \ref{lemma:sampling} below. With this result, all we need to show is that the alternative we return is in the PVC for $N$ and $M$.

    In Lines \ref{line:loop_start}--\ref{line:loop_end}, we are functionally applying the Vote by Veto rule from \citep{Moulin82} to find an element of the PVC for $N$ and $M$ using one min-query for each voter in $N$.
    Let $a$ be the alternative returned by Algorithm \ref{alg:find_epsilon_PVC_element}, and assume for contradiction that $a$ is not in the PVC for $N$ and $M$, i.e., not in the $0$-PVC for the uniform distribution over $M$. 
    Hence, there exists a coalition of voters $T$ and a set of alternatives $S$ such that $b \succ_i a$ for each $i \in T$ and $b \in S$ and such that
    \begin{equation}
    \frac{|S|}{\tau} = \frac{|S|}{|M|} > 1 - \frac{|T|}{|N|} \label{eq:contradictionlem}
    \end{equation}

    Each voter $i \in N$ removed one alternative $b_i \in M$ from the set $X$.
    Since $a$ was available in $X$ at the time and the voter chose to remove $b_i$, it must hold that $a \succ_i b_i$.
    This means that the set $W \coloneqq \{b_i \mid i \in T\}$ of alternatives removed by the voters in $T$ is disjoint from $S$.
    Moreover, $a$ is clearly not part of either set (because it was not removed and cannot be strictly preferred over itself).
    This means that $S \subseteq M \setminus W \setminus \{a\}$ and hence $|S| \leq \tau - |T| - 1$.
    It follows that
    \[ \frac{|S|}{\tau} \leq \frac{\tau - |T| - 1}{\tau} = 1 - \frac{|T|+1}{|N| + 1} \leq 1 - \frac{|T|}{|N|}, \]
    where the last inequality follows since $|T| \leq |N|$.
    This contradicts \cref{eq:contradictionlem}, showing that $a$ is indeed in the PVC for $N$ and $M$.
    Together with the lemma, we conclude that the returned statement is in the $\epsilon$-PVC with probability at least $1-\delta$.
\end{proof}

The key to concluding the proof of the theorem is \cref{lemma:sampling}, a quite technical result that connects the $\epsilon$-PVC of the full instance with the PVC for the sampled voters and alternatives.
At a high level, the proof must argue that the randomly selected set of voters and alternatives are sufficiently representative of the population and alternative distribution.
But it is not immediately obvious how to do so, given the exponential number of coalitions, and infinite collection of alternative sets that could show that an alternative is $\epsilon$-blocked.

Instead, we apply concentration inequalities in a more subtle manner so that the number of generative queries needed depends only on $\epsilon$ and is independent of $|\mathcal{N}|$ (see Line \ref{line:tau} of Algorithm \ref{alg:find_epsilon_PVC_element}). Specifically, we show that we do not need enough samples to sufficiently approximate every subset of voters and alternatives, but rather only need enough samples to sufficiently approximate the voters and alternatives for the blocking coalitions of the generated alternatives.

\begin{lemma}\label{lemma:sampling}
     Let $N$ be the set of (not necessarily distinct voters) selected in Lines \ref{line:loop_start}--\ref{line:loop_end} of Algorithm \ref{alg:find_epsilon_PVC_element}. For any $\delta, \epsilon > 0$, every element of the PVC for $N$ and $M$ from Algorithm \ref{alg:find_epsilon_PVC_element} will be in the $\epsilon$-PVC for $\mathcal{N}$ and $\mathcal{D}$.
\end{lemma}
\begin{proof}
The key observation is that every alternative $x \in \mathcal{M}$ not in the $\epsilon$-PVC for $\mathcal{M},\mathcal{N}, \mathcal{D}$ must have some blocking coalition of voters $T_x \subseteq \mathcal{N}$ and alternatives $S_x \subseteq \mathcal{M}$ satisfying Definition \ref{def:epsilon_pvc}. Denote the $i$th alternative sampled on Line \ref{line:choose_M} as $x_i$. If $x_i$ is not in the $\epsilon$-PVC for $\mathcal{M},\mathcal{N}, \mathcal{D}$, then we will show that with high probability, the corresponding blocking sets $T_{x_i}$ and $S_{x_i}$ are sufficiently accurately proportionally represented in the sampled sets $M$ and $N$ up to a factor of $\epsilon$. If this is the case, then Definition \ref{def:epsilon_pvc} implies that $x_i$ will also not be in the PVC for $M$ and $N$. For each $x_i$, define the event $E_i$ as the event that the above holds if $x_i$ is not in the $\epsilon$-PVC. If we show that $E_i$ holds with high probability for every alternative $x_i \in M$, then we can conclude that every alternative in $a \in M$ falls in one of two cases: 
Either $a$ is in the $\epsilon$-PVC for $\mathcal{M}, \mathcal{N},\mathcal{D}$ or  $a$ is not in the PVC for $N$ and $M$.

This implies the desired result that any element of the PVC for $N$ and $M$ is in the $\epsilon$-PVC  for $\mathcal{N}, \mathcal{D}$.

We now proceed with the formal proof:

    Let $X$ be a set of alternatives sampled as generative queries from $\mathcal{D}$ with size $|X| = \frac{8\log\left(\frac{32}{\epsilon^2\delta}\log\left(\frac{32}{\epsilon^2\delta}\right)\right)}{\epsilon^2}$ and let $Y \subseteq \mathcal{N}$ be a random subset of $ \mathcal{N}$ with size $|Y| = \frac{8\log\left(\frac{32}{\epsilon^2\delta}\log\left(\frac{32}{\epsilon^2\delta}\right)\right)}{\epsilon^2}$. We will show that with probability $1-\delta$, any alternative that is not in the $\epsilon$-PVC for the full set of alternatives will not be in the PVC for alternatives $X$ and voters $Y$. This implies that any element of the PVC for alternatives $X$ and voters $Y$ will with probability $1-\delta$ be in the $\epsilon$-PVC  for $\mathcal{N}$, $\mathcal{D}$. Define $\PVC{}(X,Y)$ as the PVC for alternatives $X$ and voters $Y$. 
    Define $\PVC{}_\epsilon(\mathcal{N}, \mathcal{D})$ as the $\epsilon$-PVC for  $\mathcal{N}, \mathcal{D}$.
    For any element $x \in \mathcal{M}$ such that $x \notin \PVC_\epsilon(\mathcal{N}, \mathcal{D})$, there exists some $T_x \subseteq \mathcal{N}$ and $S_x \subseteq \mathcal{M}$ such that for all $i \in T_x$ and $a \in S_x$, $a \succ_i x$ and $\mu_{\mathcal{D}}(S_x) + \frac{|T_x|}{|\mathcal{N}|} > 1+\epsilon$.

    Let $X = \{x_1,...,x_{|X|}\}$ be the ordered set of generative queries. Now we define the event $E_i := \Big\{x_i \in \PVC_\epsilon(\mathcal{N}, \mathcal{D})\Big\} 
         \bigcup \left\{x_i \not\in \PVC_\epsilon(\mathcal{N}, \mathcal{D}) \text{ and }\frac{|X \cap S_{x_i}|}{|X|} \ge \mu_{\mathcal{D}}(S_{x_i}) - \epsilon/4 \text{ and } \frac{|Y \cap T_{x_i}|}{|Y|} \ge \frac{|T_{x_i}|}{|\mathcal{N}|} - \epsilon/4\right\}$
    First, note that if $\bigcap_{i=1}^{|X|} E_i$ holds, then every element of $X$ is either in $\PVC{}_\epsilon(\mathcal{N}, \mathcal{D})$ or is not in $\PVC{}(X,Y)$. This implies that every element not in $\PVC{}_\epsilon(\mathcal{N}, \mathcal{D})$ is not in $\PVC{}(X,Y)$. 
    
    All that remains is to show that 
    $
        \Pr\left(\bigcap_{i=1}^{|X|} E_i\right) \ge 1-\delta. 
    $
    First, we will start with lower bounding $\Pr(E_i)$ (which by symmetry is the same for all $i$). We have:
    \begin{align*}
        &\Pr(E_i) = \Pr\left(x_i \in \PVC_\epsilon(\mathcal{N}, \mathcal{D})\right) + \\
        & \quad \Pr\left(\frac{|X \cap S_{x_i}|}{|X|} \ge \mu_{\mathcal{D}}(S_{x_i}) - \frac{\epsilon}{4} \text{ and } \frac{|Y \cap T_{x_i}|}{|Y|} \ge \frac{|T_{x_i}|}{|\mathcal{N}|} - \frac{\epsilon}{4} \,\middle|\, x_i \notin \PVC_\epsilon(\mathcal{N}, \mathcal{D})\right)\Pr(x_i \notin \PVC_\epsilon(\mathcal{N}, \mathcal{D})).
    \end{align*}
    Focusing on the middle probability and using Hoeffding's inequality, we have that
    \begin{align*}
        &\Pr\left(\frac{|X \cap S_{x_i}|}{|X|} \ge \mu_{\mathcal{D}}(S_{x_i}) - \frac{\epsilon}{4} \text{ and } \frac{|Y \cap T_{x_i}|}{|Y|} \ge \frac{|T_{x_i}|}{|\mathcal{N}|} - \frac{\epsilon}{4} \,\middle|\, x_i \notin \PVC_\epsilon(\mathcal{N}, \mathcal{D})\right) \\
        &\ge 1 - \Pr\left(\tfrac{|X \cap S_{x_i}|}{|X|} < \mu_{\mathcal{D}}(S_{x_i}) - \tfrac{\epsilon}{4}  \,\middle|\, x_i \notin \PVC_\epsilon(\mathcal{N}, \mathcal{D})\right) - \Pr\left(  \tfrac{|Y \cap T_{x_i}|}{|Y|} < \tfrac{|T_{x_i}|}{|\mathcal{N}|} - \tfrac{\epsilon}{4} \,\middle|\, x_i \notin \PVC_\epsilon(\mathcal{N}, \mathcal{D})\right) \\
        &\ge 1 - e^{-2|X|\epsilon^2/16}- e^{-2|Y|\epsilon^2/16}  \qquad \qquad \qquad\qquad \qquad \qquad \qquad \qquad \text{[Hoeffding's Ineq.]} \\
        &= 1 - 2e^{-2|X| \epsilon^2/16}  \qquad \qquad \qquad\qquad \qquad \qquad \qquad \qquad \qquad\qquad \qquad \qquad \text{[$|X| = |Y|$]}\\
        &= 1 - 2e^{-2\frac{8\log\left(\frac{32}{\epsilon^2\delta}\log\left(\frac{32}{\epsilon^2\delta}\right)\right)}{\epsilon^2} \epsilon^2/16} \\
        &= 1 - \frac{\delta}{\frac{16\log\left(\frac{32}{\epsilon^2\delta}\right)}{\epsilon^2}} \\
        &\ge 1 - \frac{\delta}{\frac{8\log\left(\frac{32}{\epsilon^2\delta}\log\left(\frac{32}{\epsilon^2\delta}\right)\right)}{\epsilon^2}} \\
        &= 1 - \frac{\delta}{|X|}.
    \end{align*}
    Returning to bounding $\Pr(E_i)$, we have that 
    \begin{align*}
        \Pr(E_i) &\ge \Pr\left(x_i \in \PVC_\epsilon(\mathcal{N}, \mathcal{D})\right) + \left(1 - \frac{\delta}{|X|}\right)\Pr(x_i \notin \PVC_\epsilon(\mathcal{N}, \mathcal{D})) 
        \ge 1 - \frac{\delta}{|X|}.
    \end{align*}
    Finally, we can conclude that
    $
    \Pr\left(\bigcap_{i=1}^{|X|} E_i\right)  \ge 1- \sum_{i=1}^{|X|} \Pr\left(\neg E_i\right) \ge 1-\delta
    $
    as desired.
\end{proof}

\subsection{Pairwise Discriminative Queries}\label{sec:pairwise}
While Algorithm \ref{alg:find_epsilon_PVC_element} uses min-queries,
it would be much more practically useful to have an algorithm based on pairwise discriminative queries,
which are easier for voters to respond to.
Fortunately, we can achieve this with one simple change to Algorithm \ref{alg:find_epsilon_PVC_element}.

Specifically, we can simulate each min-query in Line \ref{line:discrim} with a simple iterative algorithm to find the least-preferred alternative for a given voter in a linear number of pairwise comparisons:
First, choose an arbitrary ordering of the candidates in $X$ and do a single pairwise query asking voter $i$ to compare the first two alternatives in the ordering. Then, ask voter $i$ to compare the least-preferred alternative from that first comparison to the next alternative in the ordering. Repeat this process, always keeping the least-preferred alternative from each comparison. After iterating through all alternatives in $X$, this is guaranteed to only use $|X|-1$ pairwise queries to find voter $i$'s least favorite alternative among $X$.

\section{Lower Bounds}\label{sec:lower_bounds}
In this section, we show that Algorithm \ref{alg:find_epsilon_PVC_element} is tight in terms of both generative and discriminative queries. To prove this for each of generative and min-queries, we give a specific hard instance of the problem that requires $\tilde{\Omega}(1/\epsilon^2)$ of that type of query in order to find an element of the $\epsilon$-PVC. We then show that the number of pairwise queries in the modification of Section \ref{sec:pairwise} is also tight.

\subsection{Generative and Min-Queries}

In Theorem \ref{thm:lower_samples}, we give a lower bound on the number of generative queries needed in order to identify an element of the $\epsilon$-PVC. Importantly, this result combined with Corollary \ref{cor:suff_many} implies a fundamental gap between the number of queries needed to generate versus identify an element of the $\epsilon$-PVC. Specifically, Corollary \ref{cor:suff_many} says that $\tilde{O}(1/\epsilon)$ generative queries suffice to \emph{generate} an element of the $\epsilon$-PVC with high probability, while Lemma \ref{thm:lower_samples} says that we need at least $\tilde{\Omega}(1/\epsilon^2)$ generative queries to \emph{identify} an element of $\epsilon$-PVC. This highlights a surprising subtlety about generative queries, which is that it can be easier to stumble across a statement in the $\epsilon$-PVC than to certify which statement is in the $\epsilon$-PVC. 

\begin{theorem}\label{thm:lower_samples}
    For any $\delta, \epsilon > 0$, no algorithm can use fewer than $\Omega(\log(1/\delta)/\epsilon^2)$ generative queries and for every distribution $\mathcal{D}$ identify an element of the $\epsilon$-PVC with probability greater than $1-\delta$.
\end{theorem}
\begin{proof}
    Consider the following two discrete distributions with support $\{a_1, a_2\}$.
    Distribution $\mathcal{D}_1$ puts  $0.5 + \epsilon$ probability mass on $a_1$ and $0.5 - \epsilon$ probability mass on $a_2$. Distribution $\mathcal{D}_2$ puts  $0.5 - \epsilon$ probability mass on $a_1$ and $0.5 + \epsilon$ probability mass on $a_2$. Suppose that half of the voters in $\mathcal{N}$ have preferences $a_1 \succ a_2$ and the other half of the voters have preference $a_2 \succ a_1$. With this voter population, the $\epsilon$-PVC under distribution $\mathcal{D}_1$ is $\{a_1\}$ while the $\epsilon$-PVC under distribution $\mathcal{D}_2$ is $\{a_2\}$. Therefore, in order to find the $\epsilon$-PVC, we must have sufficient information to distinguish between distributions $\mathcal{D}_1$ and $\mathcal{D}_2$. Next, note that the KL distance between these distributions satisfies $KL(\mathcal{D}_1 || \mathcal{D}_2) \le O(\epsilon^2)$. Therefore, a standard information theory result (stated and proven in Lemma \ref{lemma:cam} for completeness) gives that with fewer than $\Omega(\log(1/\delta)/\epsilon^2)$ generative queries, it is impossible to distinguish between these two distributions with probability greater than $1-\delta$, which implies the desired result.
\end{proof}

A similar construction shows that in the worst case, we also need at least $\tilde{O}(1/\epsilon^2)$ min-queries to distinguish between two populations that are both close to ambivalent between two alternatives. The main difference between this construction and the one in Theorem \ref{thm:lower_samples} is that it is difficult to distinguish between two similar populations rather than two similar distributions.

\begin{theorem}\label{thm:min}
    For any $\delta, \epsilon > 0$, no algorithm can use fewer than $\Omega(\log(1/\delta)/\epsilon^2)$ discriminative queries and for every distribution $\mathcal{D}$ identify an element of the $\epsilon$-PVC with probability greater than $1-\delta$.
\end{theorem}

\begin{proof}
    Consider the uniform distribution $\mathcal{D}$ with support $\{a_1, a_2\}$. Consider the following two sets of voters. In $\mathcal{N}_1$, a $(1/2+\epsilon)$ fraction of voters have preference $a_1 \succ a_2$ and a $(1/2-\epsilon)$ fraction of voters have preference $a_2 \succ a_1$. In $\mathcal{N}_2$,   a $(1/2+\epsilon)$ fraction of voters have preference $a_2 \succ a_1$ and a $(1/2-\epsilon)$ fraction of voters have preference $a_1 \succ a_2$. The $\epsilon$-PVC for voters $\mathcal{N}_1$ is just $\{a_1\}$ while the $\epsilon$-PVC for the voters $\mathcal{N}_2$ is $\{a_2\}$. Distinguishing between these two voter populations with discriminative queries is equivalent to distinguishing between a $Bernoulli(1/2 + \epsilon)$ distribution and a $Bernoulli(1/2 - \epsilon)$ distribution. We can again appeal to Lemma \ref{lemma:cam} to get the desired result that no algorithm can use fewer than $\Omega(\log(1/\delta)/\epsilon^2)$ queries to identify an element of the $\epsilon$-PVC with probability higher than $1-\delta$ for both of these settings.
\end{proof}

The proofs of \Cref{thm:lower_samples} and \Cref{thm:min} both rely on the following standard result from information theory, which we state here formally and reprove for completeness.

    \begin{lemma}\label{lemma:cam}
        If two distributions $\mathcal{D}_1$ and $\mathcal{D}_2$ have Kullback-Leibler distance satisfying  $KL(\mathcal{D}_1 || \mathcal{D}_2) \le c\epsilon^2$ for some constant $c$, then no algorithm can achieve minimax misclassification error of less than $\delta$ for these two distributions using fewer than $O(\log(1/\delta)/\epsilon^2)$ samples. 
    \end{lemma}
    \begin{proof}
         Let $\mathcal{D}_1^n$ and $\mathcal{D}_2^n$ be the product distributions for $n$ samples from $\mathcal{D}_1$ and $\mathcal{D}_2$ respectively. Then by definition, $KL(\mathcal{D}_1^n, \mathcal{D}_2^n) = nc\epsilon^2$. By the Bretagnolle-Huber Inequality, this implies that the $TV(\mathcal{D}_1^n, \mathcal{D}_2^n) \le 1 - \frac{1}{2}e^{-nc\epsilon^2}$. By Le Cam's method, the minimax misclassification error for any classifier distinguishing between $\mathcal{D}_1^n, \mathcal{D}_2^n$ must be lower bounded by $\frac{1}{2}\left(1 - TV(\mathcal{D}_1^n, \mathcal{D}_2^n)\right) \ge \frac{1}{4}e^{-nc\epsilon^2}$. Plugging in $n = \log(1/(4\delta))/(c\epsilon^2)$ gives a misclassification error of at least $\delta$, which is the desired result. 
    \end{proof}
    
\subsection{Pairwise Discriminative Queries}

Theorems \ref{thm:lower_samples} and \ref{thm:min} together imply that Algorithm \ref{alg:find_epsilon_PVC_element} is tight in terms of the number of generative queries and min-queries needed. As discussed in Section \ref{sec:algorithm}, Algorithm \ref{alg:find_epsilon_PVC_element} works by finding an element of the PVC for a set of voters $N$ and alternatives $M$. We also showed that this algorithm can be modified to use $O(nm)$ pairwise queries, where $n = |N|$ and $m = |M|$. In Theorem \ref{thm:find_PVC}, we show that finding the PVC from a set of alternatives $N$ and alternatives $M$ (with sizes $n$ and $m$ respectively) cannot be done in fewer than $\Omega(nm)$ pairwise discriminative queries, which implies that there is no better subroutine than this modification in the worst case. Theorem \ref{thm:find_PVC} and its proof may be of independent interest to those studying the elicitation of voter preferences using pairwise comparisons. The proof below establishes a lower bound of $nm/32$; we present a longer proof of the same result that achieves a tighter lower bound of $nm/2$ in Appendix \ref{sec:longer_proof}. 

\begin{theorem}\label{thm:find_PVC}
	When the sets of alternatives $M$ and voters $N$ are finite, no algorithm can always find an element in the PVC using fewer than $\Omega(mn)$ pairwise discriminative queries.
\end{theorem}

\begin{proof}
    Let $m = kn+1$ for any $k \ge 1$. By definition, any alternative $a$ in the PVC must not be ranked in the last $k$ alternatives by any voter. We will show that any algorithm $ALG$ needs at least $nm/32$ queries to find any alternative not ranked in the last $k$ by any voter, which implies we need at least $nm/32$ queries to find any alternative in the PVC.

    Consider the following adversary. The adversary answers pairwise queries by $ALG$ arbitrarily (and consistently with previous queries), and for every alternative $a$ keeps track of a counter $c_a$ which counts the number of pairwise queries involving alternative $a$ up to this point. The adversary also keeps track of a set $X$ of ``active'' voters which is initialized to be $\mathcal{N}$. As soon as the counter $c_a$ for an alternative $a$ hits $n/4$, the adversary arbitrarily chooses one of the remaining ``active'' voters $i$ for which $a$ can be in the bottom $k$, and assigns $a$ to be one of voter $i$'s $k$ worst ranked alternatives. The adversary then truthfully tells $ALG$ that $a$ is ranked in one of the last $k$ spots for voter $i$. Similarly, if the queries answered so far force an alternative $a$ to be ranked in the last $k$ of a voter $i$, then $a$ is also assigned to voter $i$. If voter $i$ has $k$ distinct alternatives assigned to voter $i$'s last $k$ spots, then voter $i$ is no longer active and is removed from the set of active voters.

    Now suppose that $ALG$ has asked fewer than $nm/32$ discriminative queries. Partition the full set of alternatives $M$ into $M_1$ and $M_2$, such that $M_1$ contains all alternatives that have been involved in at least $n/4$ discriminative queries up until this point, and $M_2$ contains all alternatives that have been involved in strictly less than $n/4$ discriminative queries. Note that for any alternative in $M_1$, the adversary has already placed that alternative in the last $k$ spots of at least one voter, and therefore no alternative in $M_1$ is in the PVC.

    To conclude the proof, we must show that $ALG$ cannot determine whether or not any alternative in $M_2$ is in the PVC. Because $ALG$ has asked fewer than $nm/32$ queries so far and each query involves exactly two alternatives, the number of alternatives involved in at least $n/4$ discriminative queries is at most $m/4$. This implies that $|M_1| \leq m/4$. An alternative $a$ is only assigned to a voter if either it is in $M_1$ or if the queries for a specific voter $i$ force $a$ to be in $i$'s bottom $k$ spots. For the latter case, there must have been at least $\frac{m-1}{k} = n$ queries per such assigned $a$. Therefore, there are at least $nk - m/4 - \frac{nm/32}{n} \ge \frac{23kn - 9}{32}$ remaining spots left in the bottom $k$ positions of the active voters, which implies there are at least $23n/32 - 1$ active voters remaining. For any alternative $a \in M_2$, we know that strictly less than $n/4$ voters have been queried about $a$. Therefore, since there are at least $23n/32 -1$ active voters remaining, there must be at least $23n/32- 1 - n/4 > 0$ active voters (for sufficiently large $n$) who have not yet been queried about alternative $a$. Importantly, this means that $a$ may be in one of the last $k$ spots of these voters' rankings, which further implies that $ALG$ is unable to determine whether $a$ is in the PVC. This is true for every alternative $a \in M_2$, so we can conclude that $ALG$ cannot find any element of the PVC with fewer than $nm/32$ queries.
\end{proof}

\section{Interpretation of $\mathcal{D}$}\label{sec:distr}

In this section, we further discuss the significance of the distribution $\mathcal{D}$. One way to think about $\mathcal{D}$ is to view it as a weighting scheme that puts more weight on more ``important'' alternatives. Traditional voting theory does not explicitly assume such a weighting scheme, but does commonly assume neutrality. A voting rule satisfies neutrality if it is indifferent to the identity of the alternatives, in other words the outcome does not change if the alternatives are permuted. Therefore, any voting rule that satisfies neutrality must implicitly weight the alternatives equally, i.e. use a uniform weighting. Similarly, when there is a compact space of infinite alternatives, any voting rule satisfying neutrality must implicitly use a distribution $\mathcal{D}$ that is the uniform distribution over all alternatives.

 All of our theoretical results (both positive and negative) apply when $\mathcal{D}$ is the uniform distribution. However, when there are infinite alternatives, the uniform distribution may not always be the correct choice. As an example, consider the alternative space $\mathcal{M} = [0,1]$ with $100$ voters. $99$ of these voters only like alternative $0$ and one voter only likes alternative $1$. In this example, despite infinite alternatives existing in the interval $[0,1]$, the only two alternatives that should be relevant are $\{0,1\}$, which have measure $0$ under the uniform distribution. Unfortunately, the natural generalization of many voting rules (e.g., Borda count) ignore sets of size $0$, and therefore may even select an alternative not from $\{0,1\}$. Similarly, when computing the $\epsilon$-PVC, both of the alternatives in  $\{0,1\}$ with measure $0$ could be vetoed depending on tie-breaking. This is clearly an undesirable outcome, and can be easily avoided by choosing a more appropriate distribution $\mathcal{D}$.

Motivated by the example above, one way to choose the distribution $\mathcal{D}$ is to put more weight on alternatives that are more ``important''. This means that the distribution $\mathcal{D}$ is implicitly capturing the cardinal preferences of the voters by putting higher weight on alternatives for which voters have higher utilities. A common choice for such a distribution is the Bradley-Terry-Luce (BTL) model, where the probability of selecting an alternative is proportional to a score for that alternative. Suppose each player $i$ has a utility function $u_i : \mathcal{M} \rightarrow [0,1]$ such that $\int_{\mathcal{M}} u_i(m)dm = 1$. Then we could let $\mathcal{D}$ be the mixture of BTL models over all voters, i.e. $\mu_\mathcal{D}(m) = \frac{1}{|\mathcal{N}|}\sum_{i \in \mathcal{N}} u_i(m)$. This is equivalent to a BTL model with scores equal to the sum of voter utilities. 

Happily, choosing $\mathcal{D}$ to be this BTL model also gives us an intuitive interpretation for the $\epsilon$-PVC. If $\mathcal{D}$ is this BTL model, then the size of a coalition necessary to veto an alternative is proportional to the utilitarian social welfare (total population utility) for that alternative. This means that a larger coalition of voters is needed to veto an alternative that has higher utility among the voters. We formalize this idea below in the extreme case .

\begin{theorem}
    If $\mathcal{D}$ satisfies $\mu_{\mathcal{D}}(m) = \frac{1}{|\mathcal{N}|}\sum_{i \in \mathcal{N}} u_i(m)$ then the following holds for any subset $S \in \mathcal{M}$.  If at least an $x+\epsilon$ fraction of voters rank every alternative in $S$ above every alternative not in $S$ and the total utility of $S$ is greater than $1-x$ fraction of the total utility of all alternatives,  then the $\epsilon$-PVC must be a subset of $S$.
\end{theorem}

\begin{proof}
Let $T$ be the coalition of voters that prefers $S$ to every alternative not in $S$ and satisfies $\frac{|T|}{n} \ge x+\epsilon$. We will show that the coalition $T$ with alternatives $S$ can veto every alternative $a \in \mathcal{M} \setminus S$. This follows directly from the fact that $\mu_{\mathcal{D}}(S) = \frac{1}{|\mathcal{N}|}\sum_{i \in \mathcal{N}} u_i(m) > 1-x \ge 1 - (x+\epsilon) + \epsilon  \ge 1 - \frac{|T|}{n}+\epsilon$ which satisfies Definition \ref{def:epsilon_pvc}.
\end{proof}

In the discrete setting, we have the following corollary:
\begin{corollary}\label{axiom:weak_plurality}
    If $\mathcal{D}$ satisfies $\mu_{\mathcal{D}}(m) = \frac{1}{|\mathcal{N}|}\sum_{i \in \mathcal{N}} u_i(m)$ then the following holds. If at least an $x+\epsilon$ fraction of the population ranks alternative $a \in \mathcal{M}$ first and the utilitarian social welfare of $a$ is greater than $x$ fraction of the total utility of all alternatives, then $a$ is the only alternative in the $\epsilon$-PVC.
\end{corollary}

For intuition, suppose $x = 1/2$ and $\mathcal{D}$ is the BTL distribution (as above). In this case, Corollary \ref{axiom:weak_plurality} says that if there is any alternative $a$ such that a majority of voters rank $a$ first and $a$ has more than half of the total population utility, then $a$ is the only alternative in the $\epsilon$-PVC. Similarly, if at least a $1/2+\epsilon$ fraction of voters have utility $1$ for $a$ and utility $0$ for every other alternative, then $a$ is the only alternative in the $\epsilon$-PVC. Examples such as these demonstrate how the $\epsilon$-PVC combines both ordinal and cardinal preferences when $\mathcal{D}$ is the utility-based BTL model.

\section{Experimental Results}\label{sec:experimental-results}

We now analyze the $\epsilon$-PVC of synthetically generated preference data, modeling the preferences of individuals over opinion statements on various policy issues.
Our first experimental goal is to verify that our algorithm performs well in practical settings. This complements our asymptotic analysis in \cref{sec:algorithm}, demonstrating that our algorithm also performs well at small sample sizes. 
Our second goal
is to measure how often popular voting rules choose alternatives that are in the $\epsilon$-PVC. Our third goal is to understand how well LLMs can generate statements in the $\epsilon$-PVC given different generative capabilities and types of information. 

For each experiment, we first use the persona dataset of \citet{castricato2024persona_reproducible_testbed_pluralistic_alignment} to select the 100 distinct voters in $\mathcal{N}$.
We then generate a different alternative for each voter for a total of $100$ alternatives.
We let $\mathcal{D}$ be the uniform distribution over these alternatives, to capture a distribution 
defined by
first choosing a random voter, and then generating a statement based on that user's preferences, as discussed in \cref{sec:model}.

As in our theoretical model, the voting rules in our experiments only have access to a random subset of alternatives and voters, specifically 20 of each.
Of course, practical situations may have more than 100 voters and alternatives, and one would want to sample as many as possible, but this small-scale experiment allows us to run many experiments with LLMs at reasonable cost.
We compute the critical $\epsilon$ with respect to all $100$ alternatives and voters to measure how well the alternative selected by each voting rule does at finding common ground for the overall population.
Recall that the critical $\epsilon$ (Definition \ref{def:crit_eps}) measures the smallest $\epsilon$ for which the alternative is in the $\epsilon$-PVC.
Therefore, a method that is better at finding common ground will choose alternatives with critical $\epsilon$ close to $0$ while a method that is worse at finding common ground will choose alternatives with relatively large critical $\epsilon$.
\footnote{We open source our implementation of various \(\epsilon\)-PVC computations at \href{https://github.com/jeqcho/pvc-toolbox}{PVC Toolbox}. The pip-installable library is performance-optimized with a max-flow algorithm compiled with C through SciPy \citep{2020SciPy-NMeth}. Furthermore, we release the codebase for our experiments at this \href{https://github.com/jeqcho/single-winner-generative-social-choice}{GitHub repository}.}

\subsection{Methodology}
\begin{figure}
    \centering
    \includegraphics[width=\linewidth]{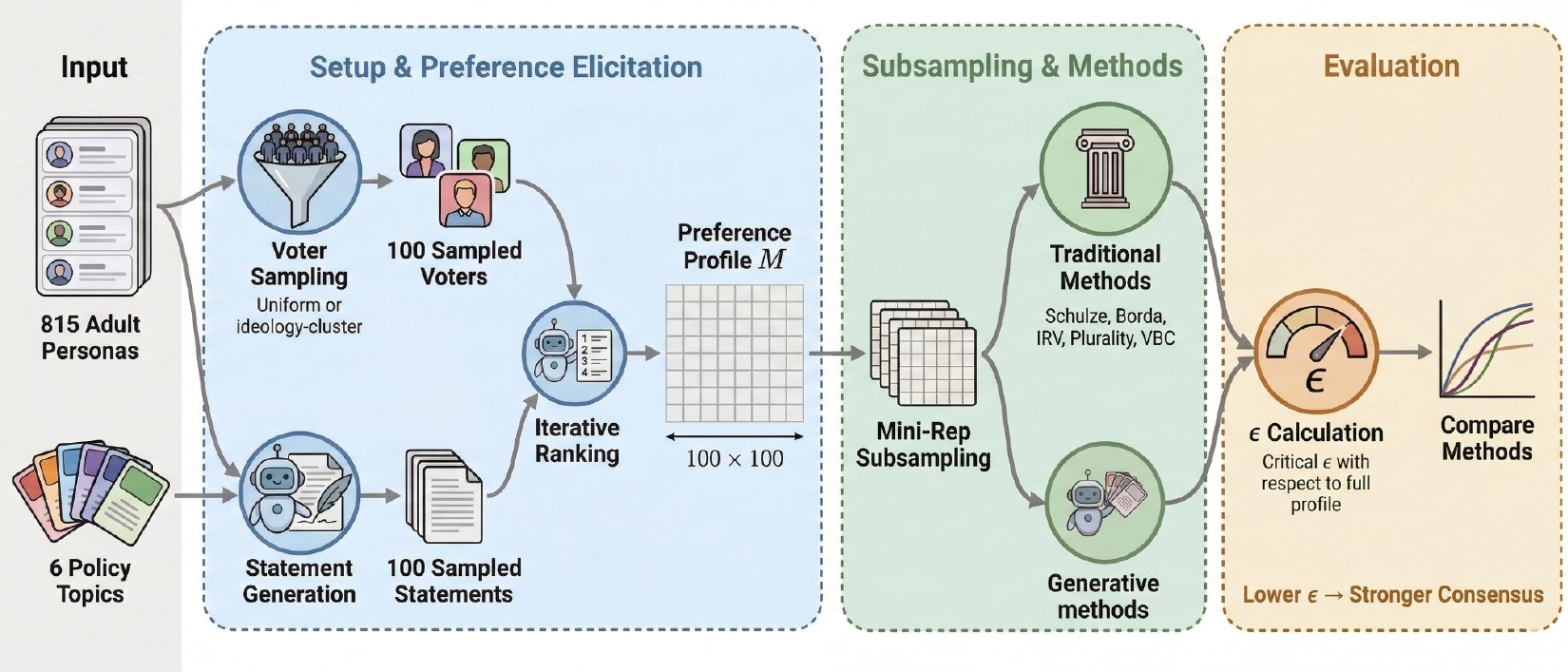}
    \caption{Our experiment setup. We generate statements conditioned on synthetic persona, and build preference profiles from voters conditioned on synthetic persona. We then evaluate various voting rules, including those powered by LLMs, and compare their critical epsilons.}
    \label{fig:experiment_flow_horizontal}
\end{figure}

Our experiment setup is illustrated in \Cref{fig:experiment_flow_horizontal}. For each of six topics, we perform 10 replications, each with 100 voters and 100 alternatives drawn from the persona dataset.
We then obtain a ranking of the generated alternatives for each voter to form the overall preference profile.
From there, we subsample to get 20 voters and 20 statements yielding a $20\times 20$ preference profile, which we do 4 times per replication. Finally, various voting rules are run on this $20\times 20$ preference profile, and we compare the critical epsilons of the selected alternatives (with respect to the original $100\times 100$ preference profile). We elaborate on the details of each step of the experiment below.

\paragraph{Topics} We originally formulated 13 topics covering contemporary political issues (\Cref{subsubsec:long-list-of-topics}). To study if certain voting rules do better or worse on polarized or less polarized questions, we rate each topic by degree of polarization.
For each topic, we generate 100 statements and sample 100 voters independently to assign a Likert score from 0-10 to each statement indicating how much they agree with that statement (\Cref{fig:likert_histograms}). We then compute the mean Likert score and take its absolute difference from 5 to measure how strongly voters feel about different topics, i.e. how polarizing topics are. Finally, we selected the top three and bottom three topics by this metric.

\paragraph{Personas} We use the 1,000 synthetic personas generated by \citet{castricato2024persona_reproducible_testbed_pluralistic_alignment}. Their choice of personas is informed by the US census and intended to be ideologically and demographically diverse. Each persona contains comprehensive personal characteristics, such as race, occupation, and how they spend their personal time. An example of a persona can be found in \Cref{subsec:example-persona}. We filter the personas to those above the age of 18, thus obtaining 815 adult personas.

\paragraph{Alternatives} For each persona and topic, we generate a representative statement from that persona for that topic by querying an LLM (specifically, \verb|gpt-5-mini|). We explore other methods to generate alternatives in Appendix \ref{subsec:alternative-methods-to-generate-alternatives}. 

\paragraph{Preference Profiles}\label{para:preferences-profile} We construct preference profiles by asking each voter (powered by \verb|gpt-5-mini| and conditioned on a persona) to sequentially rank their top 10 and bottom 10 statements until all 100 statements are ranked. We term this approach \emph{iterative ranking} and discuss further details in \Cref{subsec:building-preference-profiles}. Each preference profile consists of 100 voters and 100 statements.

\paragraph{Calculating Critical Epsilons} We compute the critical epsilons for each statement using \Cref{alg:calculate_critical} (where the critical $\epsilon$ is with respect to the full $100\times 100$ preference profile and the uniform distribution over these alternatives).

\subsection{Voting Rules}\label{sec:social_choice_fns}

In this section, we compare the critical epsilon of the statements selected by various voting rules. Veto-By-Consumption (VBC) \citep{ianovski_computing_2023}, like Algorithm \ref{alg:find_epsilon_PVC_element}, finds an element of the PVC for the subsampled preference profile.
We then compare VBC to other standard voting rules, including Borda count, Schulze \citep{Schulze11}, Instant Runoff Voting (IRV), and plurality. We also include a baseline that is the critical $\epsilon$ of a uniform randomly chosen statement among the $100$. \Cref{fig:cdf_traditional_group1} shows the cumulative distribution function (CDF) of critical epsilons across all iterations for the three controversial topics, and \Cref{tab:mean-critical-epsilon-voting-method-uniform} shows the mean critical epsilons for all six topics.
CDF graphs for the other topics and additional figures can be found at \Cref{subsec:detailed-results-epsilons}. Across topics there is a consistent pattern: the voting rules ranked in increasing order of critical epsilons are:
\begin{equation}\label{eq:vbc-beats-plurality}
    \text{VBC} < \text{Borda} < \text{Schulze} < \text{IRV} < \text{Plurality}.
\end{equation}

As expected from our theoretical results, even with access to only a small number of voters and alternatives, finding an element of the PVC for the subset gives an alternative with very low critical $\epsilon$ for the full sets. The critical $\epsilon$ for VBC is almost always $0$, with mean consistently less than $0.001$ across topics. This supports using methods for finding elements of the PVC in practical situations even without having access to full information.

Unsurprisingly, the other voting rules select alternatives that have higher critical $\epsilon$. For example, Schulze (used by the Habermas machine) selects alternatives with an order of magnitude higher critical epsilon than VBC. Perhaps more surprisingly, plurality consistently selects alternatives with higher critical $\epsilon$ than just choosing a random alternative. The social choice function that consistently does the best after VBC is Borda count. In practice, Borda count could be a good compromise as a scoring rule that is easily computable, interpretable, and also chooses alternatives with small critical $\epsilon$.

In addition to these synthetic experiments for preferences over textual statements, we also tested how well each of these voting rules perform on preference datasets from Preflib~\citep{MW13}, without subsampling, i.e., when the voting rule has full access to voters and alternatives.
In these real-world scenarios, we observe the same pattern of \Cref{eq:vbc-beats-plurality}, with Borda count consistently having the lowest critical $\epsilon$ among the standard voting rules and Plurality consistently having the highest. See \Cref{subsec:bridgingness-of-voting-methods-in-real-election-instances} for more details.

\begin{figure}[!htb]
    \centering
    \includegraphics[width=\linewidth]{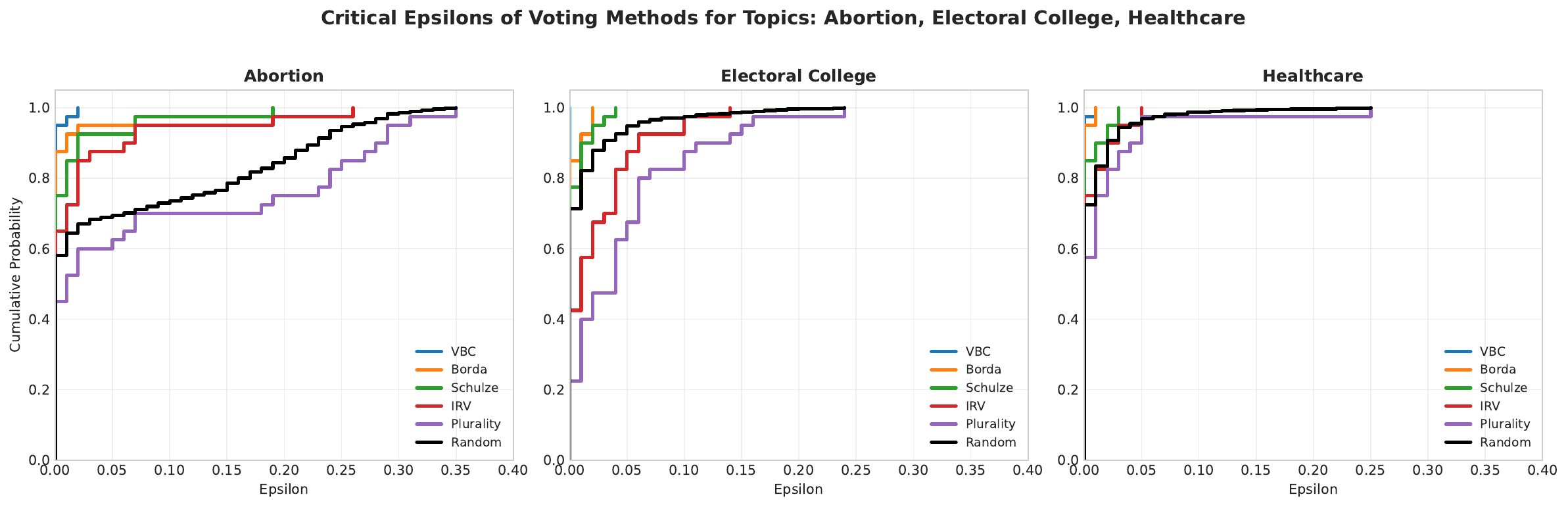}
    \caption{CDF of critical epsilons for the chosen alternatives of various voting rules for the three polarizing topics. The smaller the epsilons (the closer the line is to the top-left corner), the better at finding common ground the voting method is.}
    \label{fig:cdf_traditional_group1}
\end{figure}

\begin{table}[!htb]
\centering
\caption{Mean critical epsilons for various voting rules across various topics. The smallest epsilon for each topic is bolded.}
\label{tab:mean-critical-epsilon-voting-method-uniform}
\resizebox{\columnwidth}{!}{
    \begin{tabular}{lcccccc}
    \toprule
    Method & Abortion & Electoral College & Healthcare & Policing & Environment & Trust in Institutions \\
    \midrule
    VBC & \textbf{0.0008} & \textbf{0.0000} & \textbf{0.0003} & \textbf{0.0003} & \textbf{0.0000} & \textbf{0.0014} \\
    Borda & 0.0075 & 0.0023 & 0.0005 & 0.0015 & 0.0014 & 0.0025 \\
    Schulze & 0.0108 & 0.0040 & 0.0030 & 0.0040 & 0.0033 & 0.0033 \\
    IRV & 0.0203 & 0.0233 & 0.0063 & 0.0133 & 0.0086 & 0.0050 \\
    Plurality & 0.0865 & 0.0460 & 0.0158 & 0.0268 & 0.0142 & 0.0128 \\
\addlinespace
Random & 0.0609 & 0.0108 & 0.0082 & 0.0103 & 0.0108 & 0.0072 \\
    \bottomrule
    \end{tabular}
}
\end{table}

\subsection{Generative Voting Rules}\label{subsec:comparing-with-chatgpt-methods}

In this section we examine \emph{generative voting rules}, which do not have to select among the existing alternatives and instead can use an LLM-generated new alternative. We consider several different generative voting rules that each use different amounts of information about the problem to generate common-ground statements:
\begin{itemize}
    \item GPT-Blind: the model is only given the topic and no information about voters or alternatives
    \item GPT-Synthesize: the model is given the topic and the 20 alternatives.
    \item GPT-Synthesize+Rankings: the model is given the topic, the 20 alternatives, and the $20\times 20$ preference profile.
    \item GPT-Synthesize+Personas: the model is given the topic and the persona descriptions of the 20 voters.
\end{itemize}
To keep the values comparable, we do not compute the critical epsilon of a generated statement with respect to the uniform distribution over 101 statements.
Instead, we continue to consider the uniform distribution over the 100 original statements, and treat the generated statement as having zero probability mass in $\mathcal{D}$.

\Cref{tab:mean-critical-epsilon-generative-methods-generative} shows the mean critical $\epsilon$ for these different generative voting rules. From this table, we can see that the generative voting rules do generate alternatives with smaller critical $\epsilon$ than many standard voting rules like Schulze, IRV, and Plurality. VBC and Borda both slightly outperform the generative voting rules despite the fact that the generative methods can produce their own alternatives. Overall, the generative voting rules choose alternatives with relatively low critical $\epsilon$, indicating that what an LLM believes represents common ground is well aligned with the concept of the PVC.

\begin{table}[!htb]
\centering
\caption{Mean critical epsilons for VBC, generative voting rules and random baselines across various topics. Smallest epsilon for each topic is bolded. Under ``Random Insertion,'' we sample a random statement from the 815 global statements that are not in the $100\times 100$ preference profile and use the same insertion program that we use for the generative methods.}
\label{tab:mean-critical-epsilon-generative-methods-generative}
\resizebox{\columnwidth}{!}{
\begin{tabular}{lcccccc}
\toprule
Method & Abortion & Electoral College & Healthcare & Policing & Environment & Trust in Institutions \\
\midrule
VBC & \textbf{0.0008} & \textbf{0.0000} & \textbf{0.0003} & \textbf{0.0003} & \textbf{0.0000} & \textbf{0.0014} \\
\addlinespace
GPT-Blind & 0.0103 & 0.0093 & 0.0080 & 0.0045 & 0.0085 & 0.0200 \\
GPT-Synthesize & 0.0055 & 0.0108 & 0.0090 & 0.0128 & 0.0168 & 0.0238 \\
GPT-Synth+Rank. & 0.0030 & 0.0053 & 0.0038 & 0.0025 & 0.0080 & 0.0095 \\
GPT-Synth.+Pers. & 0.0030 & 0.0060 & 0.0040 & 0.0040 & 0.0058 & 0.0105 \\
\addlinespace
Random Insertion & 0.0700 & 0.0100 & 0.0108 & 0.0100 & 0.0073 & 0.0110 \\
Random & 0.0609 & 0.0108 & 0.0082 & 0.0103 & 0.0108 & 0.0072 \\
\bottomrule
\end{tabular}

}
\end{table}

In addition to generative LLM-based methods, we also evaluate LLM-based methods that are still restricted to choosing one of the original $100$ alternatives. We discuss these methods and results in more detail in \Cref{subsec:selector}.

\paragraph{Clustered Voter Populations}

In the previous section, we showed that the generative voting rules are able to generate alternatives with very low critical $\epsilon$, and the low numbers make it difficult to distinguish between the generative and non-generative methods. As LLMs are post-trained to be aligned to the general population, we hypothesize that generative voting rules will struggle in elections where the voter base is skewed to a particular ideology. To investigate this, we partition the 815 adult personas into 431 progressives and 255 conservatives by keyword filtering, with 129 unmatched personas (see \Cref{subsec:details-on-clustered-voters} for details). We plot our results for conservative voters in \Cref{fig:ori_cdf_generative_group1} and report the mean critical epsilon in \Cref{tab:mean-critical-epsilon-generative-methods-conservative}. We find that generative voting rules degraded significantly compared to \Cref{tab:mean-critical-epsilon-generative-methods-generative} and sometimes perform worse than the random baseline, while VBC maintains its performance of having most critical epsilons close to $0$. 

We also find that providing more information to the generative voting rules helps them significantly. Interestingly, providing personas of the voters is the most helpful to the generative voting rules, but with this information the mean critical epsilon is still 10 times larger than that of VBC. Due to space constraints, we report the results for progressive voters in Table \ref{tab:mean-critical-epsilon-generative-methods-progressive} in \cref{app:progressive}.

\begin{figure}
    \centering
    \includegraphics[width=\linewidth]{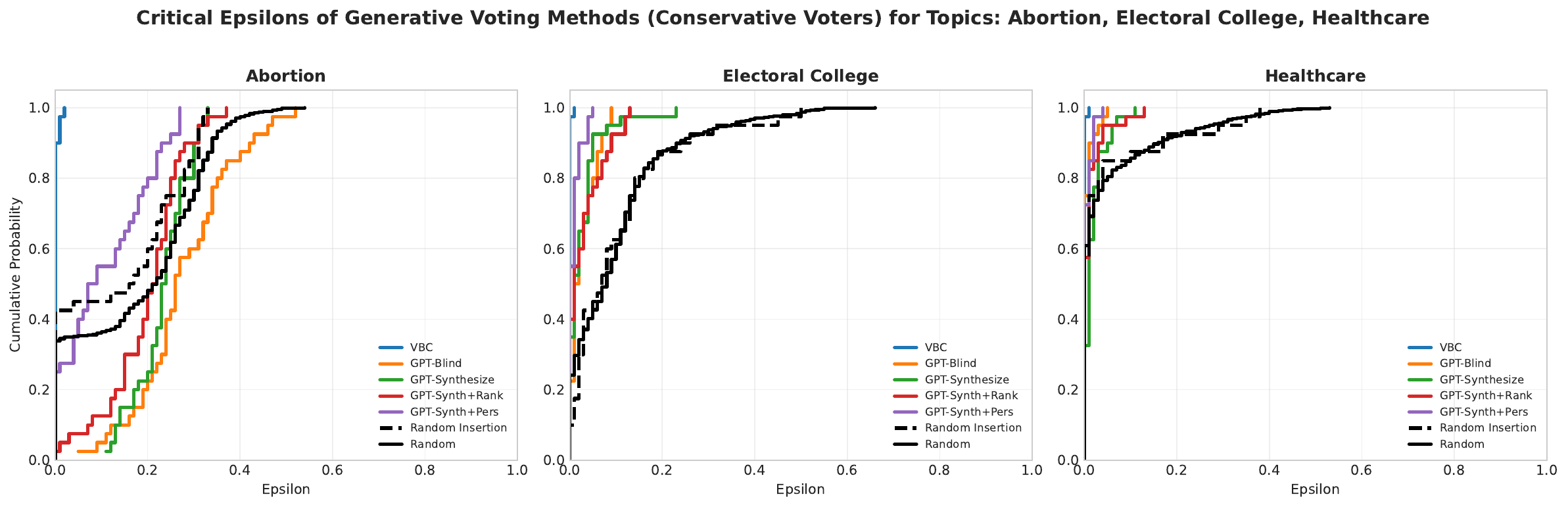}
    \caption{Critical epsilons for the winners of various generative voting rules on conservative voters for polarizing topics. 
    Generative voting rules consistently do worse than VBC and sometimes even the random baseline.
    }
    \label{fig:ori_cdf_generative_group1}
\end{figure}

\begin{table}[!htb]
\centering
\caption{Mean critical epsilons for VBC, generative voting rules and random baselines across various topics for conservative voters. Smallest epsilon for each topic is bolded.}
\label{tab:mean-critical-epsilon-generative-methods-conservative}
\resizebox{\columnwidth}{!}{
\begin{tabular}{lcccccc}
\toprule
Method & Abortion & Electoral College & Healthcare & Policing & Environment & Trust in Institutions \\
\midrule
VBC & \textbf{0.0013} & \textbf{0.0003} & \textbf{0.0003} & \textbf{0.0000} & \textbf{0.0003} & \textbf{0.0000} \\
\addlinespace
GPT-Blind & 0.2775 & 0.0263 & 0.0050 & 0.0225 & 0.0843 & 0.0220 \\
GPT-Synthesize & 0.2313 & 0.0263 & 0.0178 & 0.0393 & 0.0230 & 0.0305 \\
GPT-Synth+Rank. & 0.1980 & 0.0295 & 0.0120 & 0.0215 & 0.0280 & 0.0048 \\
GPT-Synth.+Pers. & 0.1065 & 0.0088 & 0.0048 & 0.0033 & 0.0065 & 0.0010 \\
\addlinespace
Random Insertion & 0.1360 & 0.1013 & 0.0418 & 0.0640 & 0.0153 & 0.0135 \\
Random & 0.1751 & 0.0989 & 0.0428 & 0.0625 & 0.0382 & 0.0330 \\
\bottomrule
\end{tabular}

}
\end{table}

\section{Discussion}

One of the main theoretical contributions of this paper is analyzing the sample complexity of finding alternatives in the $\epsilon$-PVC using different types of queries. While we studied certain forms of generative and discriminative queries, there are many other ways to elicit information about both alternatives and voters. A natural open question is whether there are other reasonable query models that are more effective at finding elements of the $\epsilon$-PVC. For example, different forms of discriminative or generative queries may be able to reduce the complexity of identifying an element of the $\epsilon$-PVC and close the $1/\epsilon$ to $1/\epsilon^2$ gap we showed in Section \ref{sec:lower_bounds}.

One of our key assumptions is that the number of alternatives is large or infinite, a setting in which many voting rules are not well-defined. Therefore, another open question is how to generalize other voting rules to the setting of infinite alternatives. For example, despite the PVC technically requiring a number of constraints which is exponential in $|\mathcal{N}|$, Theorem \ref{thm:alg_result} shows that the number of samples needed to find an element of the $\epsilon$-PVC does not depend on the distribution or the number of voters. One could ask whether other voting rules also have such shortcuts in the infinite alternative setting, and whether similar sample complexity results hold.

Finally, we use the critical $\epsilon$ to measure the degree to which a given statement represents common ground for the population. However, there may be many alternatives with critical $\epsilon$ equal to $0$, in which case it may not be obvious how to choose among them. One practical way to break ties would be to choose the alternative from the $\epsilon$-PVC with the highest score according to some social choice scoring function. For example, we could select the alternative with the highest Borda score among all alternatives in the $\epsilon$-PVC. If we take the Borda score as a proxy for maximizing welfare, then the Borda winner in the $\epsilon$-PVC can be interpreted as the common-ground alternative with the highest welfare. One direction for future work is to further explore the best way to combine the $\epsilon$-PVC with different scoring functions to select alternatives with small critical $\epsilon$ and high welfare.

\section*{Acknowledgements}
This work was partially supported by the National Science Foundation under grant IIS-2229881, by the Office of Naval Research under grants N00014-24-1-2704 and N00014-25-1-2153, and by grants from the Cooperative AI Foundation and the Foresight Institute. Schiffer and Zhang were supported by an NSF Graduate Research Fellowship.
We thank Teddy Lee and OpenAI for the generous support of API credits for running our experiments.

\bibliographystyle{plainnat}
\bibliography{EC26/abb,bibliography,EC26/ultimate}

\newpage

\appendix
\crefalias{section}{appendix}
\crefalias{subsection}{subappendix}      
\crefalias{subsubsection}{subsubappendix} 

\section*{\LARGE Appendix}
\section{Reduction of $\epsilon$-PVC to PVC}\label{app:pvc}

The standard definition of PVC (as in \citep{Moulin82}) is the following.
\begin{definition}\label{def:pvc}
    Let there be a set of voters $N$ with $|N| = n$, and a set of alternatives $M$ with $|M| = m$. Each voter $i$ has a ranking $\sigma_i$ over all alternatives. Let the veto function $f(x,n,m)$ for some positive integer $x \leq n$ be defined as follows:
\[
    f(x) = \left \lceil \frac{x}{n} \cdot m - 1 \right \rceil.
\]
An alternative $a$ is vetoed if there exists some group of voters $T \in N$ such that there exists a set $S$ of alternatives with $|S| \ge m - f(|T|)$ where every voter $i \in T$ prefers every alternative $b \in S$ to $a$. The set of alternatives which are not vetoed is called the \emph{proportional veto core}.
\end{definition}

Now consider Definition \ref{def:epsilon_pvc} with $\mathcal{N} = N$, $\mathcal{M} = M$, $\epsilon = 0$, and $\mathcal{D}$ as the uniform distribution over $M$. We will show this is equivalent to Definition \ref{def:pvc} Define $n = |N|$ and $m = |M|$. In this choice of parameters for Definition \ref{def:epsilon_pvc}, an alternative $a \in M$ is blocked if there exists a $T \subseteq N$ and $S \subseteq M$ such that 
\[
    \frac{|S|}{m} > 1 - \frac{|T|}{n}.
\]
Defining $y = \frac{m|T|}{n}$, this can be rewritten as
\begin{equation}\label{eq:S}
    |S| > m - y.
\end{equation}
Note that $|S| > m - y$ if and only if $|S| \ge m - \lceil y - 1 \rceil$. Therefore, Equation \eqref{eq:S} is equivalent to
\[
    |S| \ge m - \lceil y - 1 \rceil = m - f(|T|).
\]
This is exactly the same definition of being blocked as in Definition \ref{def:pvc}, and therefore we are done.

\section{Deferred Proofs}
\subsection{Proof of Proposition \ref{prop:epsilon_size}}\label{sec:epsilon_size_proof}
\begin{proof}[Proof of Proposition \ref{prop:epsilon_size}]

A key technical tool we use in the proof of Proposition \ref{prop:epsilon_size} is a linear time algorithm that, given full access to $\mathcal{D}$ and the preferences of the voters, can return a set of alternatives that are in the $\epsilon$-PVC. This algorithm is based on the Veto by Voting (VBV) algorithm in \citep{Moulin82} for finding the PVC when $M$ is finite. For any instance of the problem with finite $M$, VBV always output an element of the PVC. Our modification allows these algorithms to generalize to the setting of $\epsilon$-PVC with infinitely many alternatives. 
\begin{algorithm}[ht]
\caption{Vote by $\gamma$-Veto}
\label{alg:veto-by-consumption-text}
\DontPrintSemicolon

\textbf{Input:} A set of voters $\mathcal{N}$ and a set of alternatives $\mathcal{M}$\;

\textbf{Output:} A set of alternatives.\;

\vspace{0.5em}

Let $\gamma = \frac{1-\epsilon}{n}$.

Assign each alternatives $a \in \mathcal{M}$ capacity $c_a = \mu_{\mathcal{D}}(a)$.\;

Initialize $C  = \mathcal{M}$\;

\For{voter $i \in \mathcal{N}$}{ 
Voter $i$ identifies the smallest set $X \subseteq C$ with the following properties. First, for every $x \in X$ and $y \in C$, $y \succ_i x$. Second, defining $X_0 = \{x \in X : \mu_{\mathcal{D}}(x) = 0\}$,
\[
    \mu_{\mathcal{D}}(X_0) + \sum_{x \in X \setminus X_0} c_x \ge \gamma.
\]

Remove all alternatives in $X$ from $C$ except the one most highly ranked by voter $i$ and call the most highly ranked alternative $a$.

If $c_a \le \gamma - \mu_{\mathcal{D}}(X \setminus \{a\})$, then remove $a$ from $C$. Otherwise, let $c_a = c_a - (\gamma - \mu_{\mathcal{D}}(x \setminus \{a\}))$.
}

\Return{ $C$}
\end{algorithm}

Proposition \ref{prop:epsilon_size} follows directly from the following lemma analyzing Algorithm \ref{alg:veto-by-consumption-text}.
\begin{lemma}\label{lemma:alg_analysis}
    Every alternative in the set $C$ returned by Algorithm $\ref{alg:veto-by-consumption-text}$ is in the $\epsilon$-PVC. Furthermore, $\mu_{\mathcal{D}}(C) \ge \epsilon$.
\end{lemma}
\begin{proof}
    We first will show that the set of alternatives returned by Algorithm \ref{alg:veto-by-consumption-text} must all be in the $\epsilon$-PVC.  Proof by contradiction. Suppose that $a$ is returned by Algorithm \ref{alg:veto-by-consumption-text} and that $a$ is not in the $\epsilon$-PVC. By definition of $\epsilon$-PVC, there must exist a coalition $T$ and set of alternatives $S$ such that $w \succ_i a $ for all $i \in T$ and $w \in S$ and 
    $
    \mu_{\mathcal{D}}(S) > 1 - \frac{|T|}{n} + \epsilon.
    $
    For $i \in T$, let $W_i$ be the set of alternatives that $i$ ranks lower than $a$, and let $W = \{a\} \cup \left(\cup_{i \in T} W_i\right)$. By construction, we must have that $S$ and $W$ are disjoint, so
    \[
        \mu_{\mathcal{D}}(W) \le  1 - \mu_{\mathcal{D}}(S) < \frac{|T|}{n} - \epsilon.
    \]
    Because $a$ is in $C$ when the algorithm terminates, every voter in $T$ must have removed or decreased the capacity of alternatives they like less than or equal to $a$, so every voter in $T$ must have removed or decreased the capacity of alternatives in $W$. Each voter in $S$ consumes $\gamma$ mass, therefore the total amount of mass that is consumed from $W$ must be at least
    \[
        \gamma |T|  = \frac{|T|(1-\epsilon)}{n} = \frac{|T|}{n} - \frac{|T|}{n}\epsilon \ge \frac{|T|}{n} - \epsilon > \mu_{\mathcal{D}}(W).
    \]
    Therefore, all alternatives in $W$ must have been removed from $C$ by the time that the algorithm terminated, and therefore $a$ must have been removed from $C$. This is a contradiction with the fact that $a$ is returned by the algorithm, which completes the first part of the proof.

    To see that $\mu_{\mathcal{D}}(C) \ge \epsilon$, define 
    \[
        m(C) := \mu_{\mathcal{D}}(\{x \in C : \mu_{\mathcal{D}}(x) = 0\}) + \sum_{x \in C : \mu_{\mathcal{D}}(x) \ne 0 } c_x.
    \]
    By construction, $m(C)$ decreases by exactly $\gamma$ in each round of the for loop.  
    Since $\gamma = \frac{1-\epsilon}{n}$ and there are $n$ rounds of the for loop, this implies that $m(C) = \epsilon$ when the algorithm terminates. Furthermore, by definition $m(C) \le \mu_{\mathcal{D}}(C)$, therefore $m(C) = \epsilon$ implies that $\mu_{\mathcal{D}}(C) \ge \epsilon$ as desired.
\end{proof}
\end{proof}

\subsection{Alternative Proof of Theorem \ref{thm:find_PVC}}\label{sec:longer_proof}

In this section, we present an alternative proof for Theorem \ref{thm:find_PVC} that achieves a tighter lower bound of $mn/2$ (improvement over $mn/32$ as shown in the proof in Section \ref{sec:lower_bounds}). Note that our algorithm using Median of Medians only uses $\approx 3.33mn$ pairwise discriminative queries, and therefore this tighter lower bound is only a factor of $\approx 6.66$ from being fully tight.

\begin{proof}[Second Proof of Theorem \ref{thm:find_PVC}]
Let $m = n+1$.

First, assume that we know we are in one of two cases:
\begin{itemize}
	\item Alternative $a$ is ranked second to last in every single ranking, and every alternative in $A \setminus \{a\}$ is ranked last in exactly one ranking
	\item Alternative $a$ is ranked second to last by $n-1$ voters and is ranked last in exactly one ranking, in which some alternative $x$ is ranked second to last. Every alternative in $A \setminus \{a,x\}$ is ranked last in exactly one ranking.
\end{itemize}

In the first option above, the PVC is just $\{a\}$. In the second option above, the PVC is just $\{x\}$. Therefore, in order to find the proportional veto core, the algorithm must be able to distinguish between these two cases. 

We will show that no algorithm can distinguish between these two cases with fewer than $\frac{nm}{2} = \frac{n(n+1)}{2}$ queries by giving an adversary that requires at least this many queries.

Define $X = A \setminus \{a\}$. We therefore will consider the easier problem of determining if every element of $X$ is ranked last exactly once in $N$. Note that any queries not involving $a$ will not give any information for this alternate problem, therefore we will only consider queries involving $a$.

We will prove the desired result using induction. Let $Q$ be a set of queries of the form $(x,a, y)$ for $x \in X$ and $y \in N$ where $x \succ_y a$. For any $X,N,Q$ such that $|X| = |N|$, define $f(X,N,Q)$ as the minimum number of queries in this adversarial setting necessary to determine whether or not every element in $X$ is ranked last in exactly one ranking in $N$ assuming that we start by doing the queries in $Q$ .

Base case: If $|X| = |N| = 1$ then $X$ and $N$ are both singletons (i.e. $X = \{x\}$ and $N = \{y\}$). If $|Q| \ge 1$, then clearly $f(X,N,Q) \ge |Q| \ge 1$. If $Q = \emptyset$, then we then need one query of $(x,a,y)$ to determine whether or not $x$ is ranked last by $y$. Therefore, for any $X,N$ such that $|X| = |N| = 1$ and any $Q$, we have that
$	f(X,N,Q) \ge 1 = \frac{|N|(|N|+1)}{2}.
$
Inductive Hypothesis: Assume that $f(X,N,Q) \ge \frac{|N|(|N|+1)}{2}$ for all $X,N,Q$ such that $|X| = |N| < n$.

Consider any $X,N,Q$ such that $|X| = |N| = n$.

Case 1: 
 There exists a set $X_1 \subset X$ and $Y_1 \subset N$ such that $|X_1| = |Y_1| < n$ and such that the only possible way to assign every alternative in $X$ to be last in some ranking in $N$ is to assign every alternative in $X_1$ to be last in some ranking in $Y_1$. Define $Y_2 = N \setminus Y_1$ and $X_2 = X \setminus X_1$. Let $Q_1 \subset Q$ be the set of all queries in $Q$ involving alternatives in $X_1$ and voters in $Y_1$ and let $Q_2 \subset Q$ be the set of all queries in $Q$ involving alternatives in $X_2$ and voters in $Y_2$. We will show that
\begin{equation}\label{eq:f_X_Y_Q}
	f(X,Y,Q) \ge f(X_1,Y_1,Q_1) + f(X_2,Y_2,Q_2) + |X_1| \cdot (n - |X_1|).
\end{equation}
To see this result, note that since $X_1$ alternatives must be ranked last in some ranking in $Y_1$ and alternatives in $X_2$ must be ranked last in some ranking in $Y_2$, from now on the player only needs to ask queries involving alternatives in $X_1$ and voters in $Y_1$ or ask queries involving alternatives in $X_2$ and voters in $Y_2$. This explains the first two recursive terms in the above equation. To show the last additive term, we must show that the number of queries in $Q \setminus Q_1 \setminus Q_2$ ( or equivalently the number of queries involving alternatives in $X_1$ and voters in $Y_2$ or involving alternatives in $X_2$ and voters in $Y_1$) is at least $|X_1| \cdot (n - |X_1|)$.

This is because in order for such an $X_1$ and $Y_a$ to exist, we must be in one of the following two situations:
\begin{itemize}
	\item Every alternative  $x_1 \in X_1$ has been compared to $a$ in every ranking $y_2 \in Y_2$ and $x_1 \succ_{y_2} a$
	\item Every alternative  $x_2 \in X_2$ has been compared to $a$ in every ranking $y_1 \in Y_1$ and $x_2 \succ_{y_1} a$
\end{itemize}

Option 1 above requires $|X_1||Y_2|$ comparisons in $Q \setminus Q_1 \setminus Q_2$ and option 2 requires $|X_2||Y_1|$ comparisons in $Q \setminus Q_1 \setminus Q_2$. In both cases, we must have that
\[
	|Q \setminus Q_1 \setminus Q_2| \ge  |X_1| \cdot (n - |X_1|).
\]

Therefore, we have shown Equation \eqref{eq:f_X_Y_Q}.

Using the inductive hypothesis, Equation \eqref{eq:f_X_Y_Q} implies that
\begin{align*}
f(X,N,Q) 
&\ge \frac{|X_1|(|X_1| +1)}{2} + \frac{|X_2|(|X_2| + 1)}{2} + |X_1|\,(n - |X_1|) \\
&= \frac{|X_1|(|X_1| +1)}{2} + \frac{|X_2|(|X_2| + 1)}{2} + |X_1|\,|X_2| \tag{$|X_2| = n - |X_1|$} \\
&= \frac{|X_1|^2 + |X_1| + |X_2|^2 + |X_2| + 2|X_1||X_2|}{2} \\
&= \frac{(|X_1| + |X_2|)^2 + (|X_1| + |X_2|)}{2} \\
&= \frac{|N|^2 + |N|}{2} \\
&= \frac{|N|(|N|+1)}{2}.
\end{align*}

Case 2: If we are not in case 1, then for any query of the form $q = (x,a,y)$, suppose that the adversary responds that $x \succ_y a$. In this case, we have that (for $Q' = Q + \{q\}$)
\[
	f(X,Y,Q) = f(X,Y,Q').
\]

There are a finite number of potential queries, so eventually we will reach some $Q'$ where $Q \subseteq Q'$ and such that $f(X,Y,Q) = f(X,Y,Q')$ and such that $X,Y,Q'$ falls in case 1. We already showed that in Case 1 we must have that $f(X,Y,Q') \ge \frac{n(n+1)}{2}$, so we can conclude that
\[
	f(X,Y,Q) = f(X,Y, Q') \ge \frac{n(n+1)}{2}.
\]

In order to show that the above adversary is valid, we must maintain that at any point given the query feedback from the adversary, $a$ could either never ranked last or $a$ could be ranked last exactly once. Because the adversary always responds that $x \succ_y a$ for all queries, at any point it is still possible that $a$ could be ranked last by every voter. Therefore, we must show that for this adversary, we also could have that $a$ is not ranked last by one voter. This will follow from the following lemma. This lemma implies that any time one query response from the adversary would result in a query set $Q$ for which $a$ must be ranked last by every voter in $N$, we must be in the setting of case 1. This implies that the adversary never needs to answer a query in such a way that forces $a$ to be ranked last by every voter until at least $n(n+1)/2$ queries have already been asked.

\begin{lemma}
	For a given $(X,N,Q)$, suppose that there exists a matching of alternatives in $X$ to voters in $N$ such that every alternative in $X$ could be ranked last by exactly one voter in $N$ given the results of the queries in $Q$. Further, assume that there exists a query $(a, x,y)$ such that if $x \succ_y a$, then this is no longer the case. Then there must exist a set $X_1 \subset X$ and $Y_1 \subset Y$ such that $|X_1| = |Y_1| < n$ and such that the only possible way to assign every alternative in $X$ to be last in some ranking in $N$ is to assign every alternative in $X_1$ to be last in some ranking in $Y_1$. 
\end{lemma}

\begin{proof}

 Because there exists a matching of alternatives in $X$ to voters in $N$, choose one such mapping and call it $M : X \to N$. Create the following directed graph with vertices corresponding to $X$. Add a directed edge from $x_1 \in X$ to $x_2 \in X$ if and only if the alternative $x_1$ could be ranked last by the voter $M(x_2)$. If there exists an $(x,y)$ as assumed in the lemma, then there cannot be a directed cycle containing the vertex $x$. This is because if there is a directed cycle $C$ containing $x$, then there still exists a valid matching $M'$ that satisfies $x \succ_y a$, which violates the lemma assumption. To construct $M'$, we could simply modify $M$ by matching ever vertex $b$ in the cycle to $M(C(b))$.

 Now we show that the lack of cycles containing $x$ implies that there exists some subset of vertices $X_1 \subset X$ such that there are no edges from $X \setminus X_1$ to $X_1$ and $X_1 \subset X$. The simplest way to do this is to reduce the graph into strongly connected components. Because there is no cycle containing $x$, $x$ must be alone in its component. The strongly connected component graph is a DAG so it must have a source, which is a component with no incoming edges. Furthermore, because $x$ is alone in its component there are at least two components. Therefore, taking $X_1$ to be the source component gives  that there are no edges from $X \setminus X_1$ to $X_1$. Taking $Y_1 = M(X_1)$ gives that $X_1$ and $Y_1$ satisfy the desired property that the only possible way to assign every alternative in $X$ to be last in some ranking in $N$ is to assign every alternative in $X_1$ to $Y_1$, so we are done.

\end{proof}

\end{proof}

\section{Deferred Proofs for Computing \(\epsilon\)-PVC and Critical Epsilons}
\label{sec:deferred-proofs-for-eps-pvc-critical-epsilon}
\subsection{Computing the \(\epsilon-\)PVC and Critical Epsilons}
\label{subsec:computing-eps-pvc-and-critical-epsilons}

In this section, we first discuss computing the $\epsilon$-PVC efficiently when we have access to full information about voters and alternatives when $\mathcal{M}$ is finite. \citet{ianovski_computing_2023} showed that for the original definition of PVC, the problem of finding all statements in the PVC can be reduced to a biclique problem, for which a polynomial time algorithm is known. We adapt the proof to show that a similar reasoning applies to \(\epsilon\)-PVC.  

Fix an alternative \(a\in \mathcal{M}\). As we will use a graph reduction, we first present the definition of a \emph{blocking graph}.

\begin{definition}[Blocking graph][\citet{ianovski_computing_2023}]\label{def:blocking-graph}
    Given an election instance \((\mathcal{N},\mathcal{M},\succ)\) and an alternative \(a\in \mathcal{M}\). The \emph{blocking graph} of \(a\) is a bipartite graph where there are \(nm\) left vertices, where each vertex can be identified as one of the \(m\) copies of the \(n\) voters, and there are \(n(m-1)\) right vertices, where each vertex can be identified as one of the \(n\) copies of the \(m-1\) alternatives of \(\mathcal{M}\setminus\{a\}\). An edge exists from a left vertex to a right vertex if and only if the left vertex corresponds to a voter \(i\in \mathcal{N}\) and the right vertex corresponds to an alternative \(c\in \mathcal{M}\setminus\{a\}\) such that \(i\) prefers \(c\) to \(a\).
\end{definition}

We now present two statements and show that they are equivalent:

\begin{condition}[Blocking Set]\label{st:blocking-set}
    There exists a coalition \(T\subseteq \mathcal{N}\) and a blocking set \(S\subseteq \mathcal{M}\) such that \(|T|/n>1-|S|/m+\epsilon\) and every voter in \(T\) prefers every alternative in \(S\) to \(a\).
\end{condition}

\begin{condition}[Blocking Graph]\label{st:blocking-graph}
    The blocking graph of \(a\) has a biclique with \(mk\) left vertices and \(nb\) right vertices such that \(mk+nb> (1+\epsilon)nm\).
\end{condition}

To bridge these two statements, we use \(|T|=k\) and \(|S|=b\). By showing that \Cref{st:blocking-set} is equivalent to \Cref{st:blocking-graph}, or in words, that the existence of the blocking set is equivalent to the existence of a sufficiently large biclique in the blocking graph, we can use a known polynomial time algorithm for \Cref{st:blocking-graph} to check \Cref{st:blocking-set}, which in turn determines membership in \(\epsilon-\)PVC.

\begin{proposition}\label{prop:blocking-set-blocking-graph-equivalence}
    \Cref{st:blocking-set} is equivalent to \Cref{st:blocking-graph}. 
\end{proposition}

\begin{proof}
    (\Cref{st:blocking-set}\(\implies\)\Cref{st:blocking-graph}) Note that by multiplying the inequality in \cref{st:blocking-graph} by \(nm\), we get \(mk+nb>(1+\epsilon)nm\). By construction of the blocking graph, the \(mk\) vertices on the left that correspond to the \(m\) copies each of the \(k\) voters in \(T\) are all connected to the \(nb\) vertices on the right that correspond to the \(n\) copies each of the \(b\) alternatives in \(S\).
    
    (\Cref{st:blocking-graph}\(\impliedby\)\Cref{st:blocking-set}) Take one left vertex \(x\) in the biclique. Say \(x\) corresponds to a copy of voter \(i\). Enlarge the biclique by including all left vertices that correspond to copies of voter \(i\). This is possible because these vertices have the same set of neighbours by construction of the blocking graph. Do the same for every left vertex in the biclique. Similarly, do a corresponding enlargement for the right vertices in the biclique. Now, since there are at least \(mk\) left vertices, of which we included all \(m\) copies of their corresponding voters, we must have at least \(k\) voters included in the biclique. Similarly, as we have at least \(nb\) right vertices, of which we included all \(n\) copies of their corresponding alternatives, we must have at least \(b\) alternatives included in the biclique. Let these \(k\) voters be \(T\) and these \(b\) alternatives be \(S\). Then note that \(mk+nb>(1+\epsilon)nm\implies m|T|+b|S|>(1+\epsilon)nm\implies |T|/n>1-|S|/m+\epsilon\). By construction of the blocking graph, these voters in \(T\) all prefers every alternative in \(S\) to \(a\).
\end{proof}

Prior work has established that \cref{st:blocking-graph} can be checked in polynomial time:

\begin{proposition}[][\citet{garey_computers_1990}, pg. 196, GT24]\label{prop:garey}
    Let \(G\) be a bipartite graph and \(k\) an integer. We can, in polynomial time, determine whether there exists a biclique \(K\subseteq G\) with \(x\) left vertices and \(y\) right vertices such that \(x+y=k\).
\end{proposition}

We can thus determine whether a given alternative is in the \(\epsilon-\)PVC by using the polynomial time algorithm in \cref{prop:garey} and iterating \(k\) from \((1+\epsilon)nm\) to \(2mn-n\), giving us \Cref{prop:eps-pvc-polytime}. 

\begin{corollary}\label{prop:eps-pvc-polytime}
The \(\epsilon-\)PVC can be computed in polynomial time.
\end{corollary}

\subsection{A Faster Algorithm through a Correspondence with Min-Cut}
\label{subsubsec:algo-pvc-membership}

With proof of existence of a polynomial time algorithm out of the way, here we proceed to present the proof and runtime analysis for the faster algorithm that is used in our implementation (\Cref{alg:calculate_critical}). Many ideas for this algorithm, including the max-flow correspondence, are from \citet{ianovski_computing_2023}'s \(O(m\max(n^3,m^3))\) algorithm for computing the PVC. The core of the idea is that instead of copies of the same vertex in a blocking graph, we turn these copies into weighted edges in a flow graph where the weight is the number of copies.

\subsubsection{Correctness}\label{subsubsec:min-cut-correctness}
\begin{proof}
Consider the flow graph after applying the min cut. The min cut could not cut the edges between the voter nodes and the alternative nodes as the edge weights are infinite, so the only edges possibly cut are those between the source node and the voter nodes, and those between the alternative nodes and the sink node. Note that whenever an edge is cut, that node contributes to the min-cut, while the remaining nodes after the min-cut forms a biclique of the blocking graph, which is in fact maximum. To be precise, all voter nodes and alternative nodes are partitioned into either contributing to the min-cut or the biclique. When an edge between the source and a voter node is cut, that voter node contributes \(m\) to the min-cut. When an edge between the sink and an alternative node is cut, that alternative node contributes \(n\) to the min-cut. The remaining nodes after the min-cut will form a biclique in the corresponding blocking graph. To see this, note that since the remaining voter nodes are still connected to the source, and the remaining alternative nodes are still connected to the sink, there must not be any edges left between any remaining voter nodes and alternative nodes. Otherwise, a single edge between a voter node and an alternative node, which has infinite edge weight, could transmit positive flow.\footnote{The careful reader will find that the infinite edge will not transmit flow if all voter nodes are cut from the source, or all alternative nodes are cut from the sink. Indeed this is the baseline case, which is to cut all alternative nodes from the sink for a cut of \((m-1)n\) (which is less than cutting voter nodes from the source for \(mn\)), and will yield \(Q=mn\implies Q\leq (1+\epsilon)nm\;\forall\epsilon\in[0,1]\) (here \(Q\) loses its meaning as the maximum biclique, but the algorithm is still valid). By finding smaller cuts than \(mn\), the algorithm will obtain a bigger corresponding blocking set, and raise \(\epsilon\).} Since no edges exist between the remaining voter nodes and alternative nodes, by construction, all of these voters must approve all of these alternatives over the given alternative, which is precisely the blocking set, or equivalently, the maximum biclique in the blocking graph. Now, the size of this biclique in the blocking graph involves the number of copies, which is precisely the weights of the edges leftover from the min-cut. Since the min-cut minimizes the weights cut, the constructed biclique is the biggest possible. To be precise, \(Q+K=2mn-n\), where \(Q\) is the size of the biclique, \(K\) is the min-cut, and \(2mn-n\) is the total weight of the flow graph (excluding infinities), or equivalently, the total number of vertices in the blocking graph. With \(Q\) obtained, the correspondence in \Cref{prop:blocking-set-blocking-graph-equivalence} where \(Q=mk+nb\) allows us to check \(\epsilon-\)PVC membership with \(Q\leq (1+\epsilon)nm\) as desired.
\end{proof}

\subsubsection{Runtime}\label{subsubsec:critical-epsilon-runtime}
The runtime of the algorithm is chiefly that of the construction of the flow graph, and finding the min-cut. The flow graph construction takes \(O(nm)\) time, since the flow graph could be fully connected. For example, consider the case where the given alternative is ranked first by all voters. The min-cut algorithm depends on implementation. By the max-flow min-cut theorem, we could use standard max-flow algorithms. Dinic's algorithm \citep{Dinitz1970} requires \(O(|V|^2|E|)\), which in our case where \(|V|=O(m+n)\) and \(|E|=O(mn)\), the algorithm runs in \(O(mn(m+n)^2)\).

\section{Additional Experiment Details}
We present additional details for the experiments described in \Cref{sec:experimental-results}.
\subsection{Details on Topics}
Below we present the long list of topics and their degree of polarization.

\subsubsection{Long List of Topics}\label{subsubsec:long-list-of-topics}
\begin{itemize}
    \item How should we increase the general public's trust in US elections?
    \item What are the best policies to prevent littering in public spaces?
    \item What are your thoughts on the way university campus administrators should approach the issue of Israel/Gaza demonstrations?
    \item What should guide laws concerning abortion?
    \item What balance should exist between gun safety laws and Second Amendment rights?
    \item What role should the government play in ensuring universal access to healthcare?
    \item What balance should be struck between environmental protection and economic growth in climate policy?
    \item What principles should guide immigration policy and the path to citizenship?
    \item What limits, if any, should exist on free speech regarding hate speech?
    \item What responsibilities should tech companies have when collecting and monetizing user data?
    \item What role should artificial intelligence play in society, and how should its risks be governed?
    \item What reforms, if any, should replace or modify the Electoral College?
    \item What strategies should guide policing to address bias and use-of-force concerns while maintaining public safety?
\end{itemize}

The Likert scores for each topic used to calculate degree of polarization can be found in \Cref{fig:likert_histograms}.
\begin{figure}
    \centering
    \includegraphics[width=.5\linewidth]{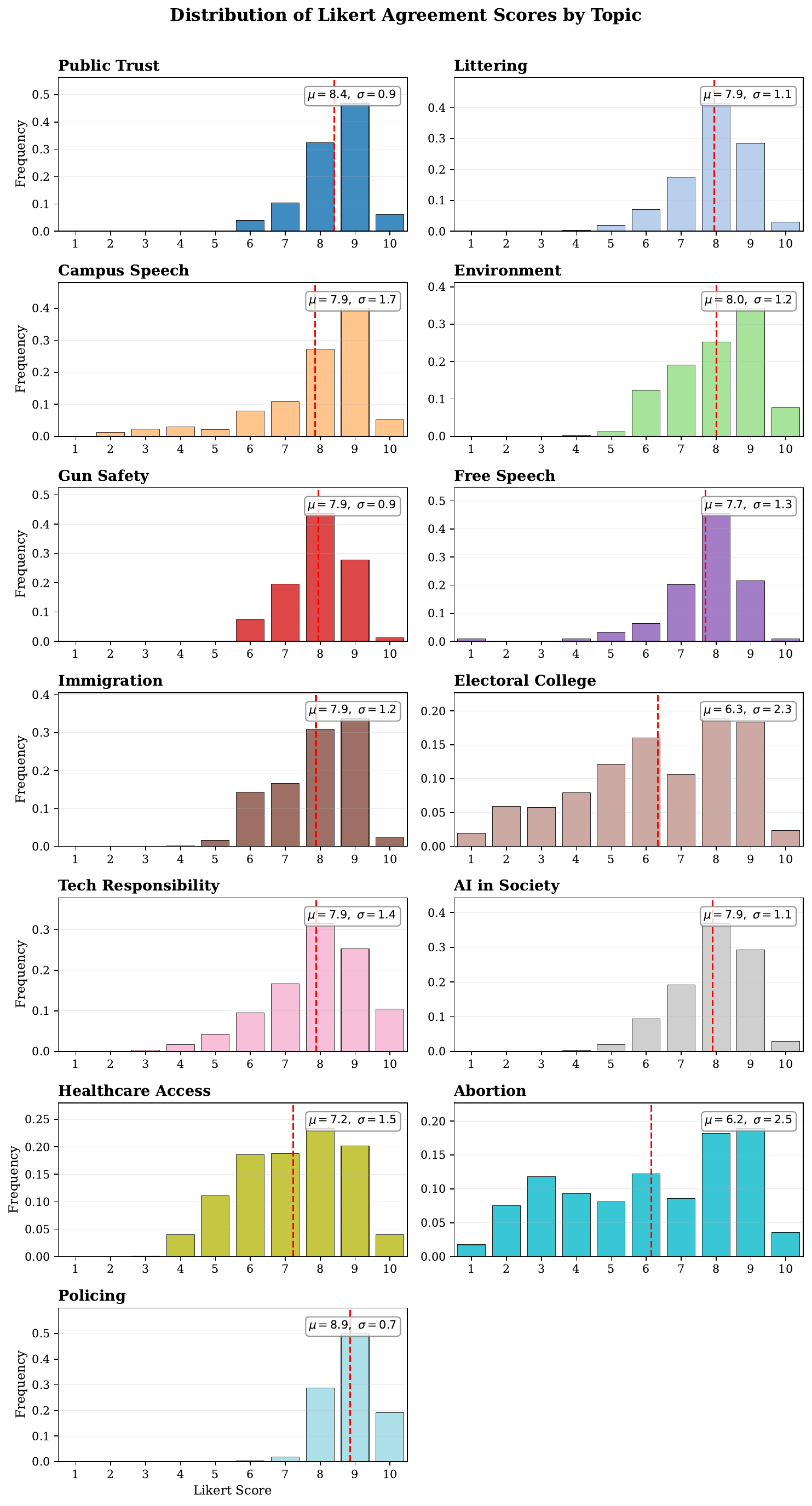}
    \caption{Likert scores of 100 voters on 100 generated statements. Higher Likert scores indicate higher agreement. The most polarizing topics have average Likert scores close to 5 and a wide spread. The top 3 topics by polarization are abortion, Electoral College, and healthcare. The bottom 3 polarizing topics are environment, policing, and public trust in elections.}
    \label{fig:likert_histograms}
\end{figure}

\newpage

\subsection{Example Persona}\label{subsec:example-persona}
In this section we present an example persona that we use as a voter and for generating statements.

\begin{promptbox}{Example Persona}
age: 72
sex: Female
race: White alone
ancestry: Irish
household language: English only
education: Bachelor's degree
employment status: Civilian employed, at work
class of worker: Employee of a private for-profit company or business, or of an individual, for wages, salary, or commissions
industry category: MED-Nursing Care Facilities (Skilled Nursing Facilities)
occupation category: MED-Registered Nurses
detailed job description: Provides patient care in a nursing home
income: 135500.0
marital status: Widowed
household type: Cohabiting couple household, no children of the householder less than 18
family presence and age: No family under 18
place of birth: Missouri/MO
citizenship: Born in the United States
veteran status: Non-Veteran
disability: None
health insurance: With health insurance coverage
big five scores: Openness: Average, Conscientiousness: Average, Extraversion: Average, Agreeableness: Extremely High, Neuroticism: Low
defining quirks: Collects vintage medical equipment
mannerisms: Speaks with a soft, soothing voice
personal time: Spends free time gardening or reading
lifestyle: Active and community-oriented
ideology: Conservative
political views: Republican
religion: Catholic
\end{promptbox}

\subsection{Alternative Methods to Generate Alternatives}\label{subsec:alternative-methods-to-generate-alternatives}
\subsubsection{Description of Methods to Generate Alternatives}\label{sssec:description-of-methods-to-generate-alternatives}

In addition to the experiment set-up presented in the body (listed below as Persona Only), we also explored three  different methods of generating statements.
\begin{itemize}
    \item \textbf{Blind}: Given only the topic, the LLM is queried using verbalized sampling \citep{zhang2025verbalizedsamplingmitigatemode} to generate diverse statements.
    \item \textbf{Persona Only}: Given the topic and a persona to conditioned on, the LLM is queried to produce a response. This is used for the experiments presented in the body.
    \item \textbf{Deliberation Round Only}: Within a rep, given 100 statements generated using Persona, the LLM is queried using verbalized sampling \citep{zhang2025verbalizedsamplingmitigatemode} to generate diverse statements taking into account the 100 statements. This is to simulate a "deliberation round" where a user sees statements written by others and is informed by their stances to write a new statement.
    \item \textbf{Persona+Deliberation Round}: Within a rep, first generate 100 statements using Personas. Then conditioning on each persona, the LLM is queried to produce a single statement that also takes into account the other 99 statements.
\end{itemize}

In \Cref{fig:combined_polarization_plots_vertical_top_three,fig:combined_polarization_plots_vertical_bottom_three} we show the CDFs for the same experiment as in the body but using different methods of generating alternatives. Using the Persona method gives a more interesting distribution of alternatives and has the added benefits of being relatively simple and motivated by the theory in Section \ref{sec:distr}. 

\begin{figure}[htbp]
    \centering
    \begin{subfigure}[b]{0.9\textwidth}
        \centering
        \includegraphics[width=\linewidth]{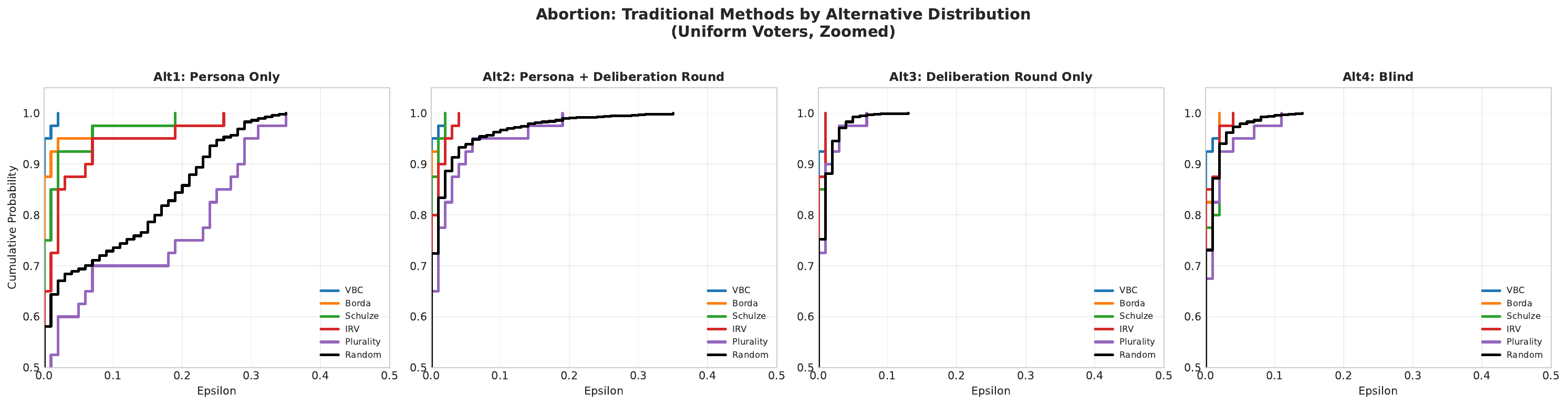}
        \caption{Abortion}
        \label{fig:cdf_abortion}
    \end{subfigure}

    \vspace{10pt}

    \begin{subfigure}[b]{0.9\textwidth}
        \centering
        \includegraphics[width=\linewidth]{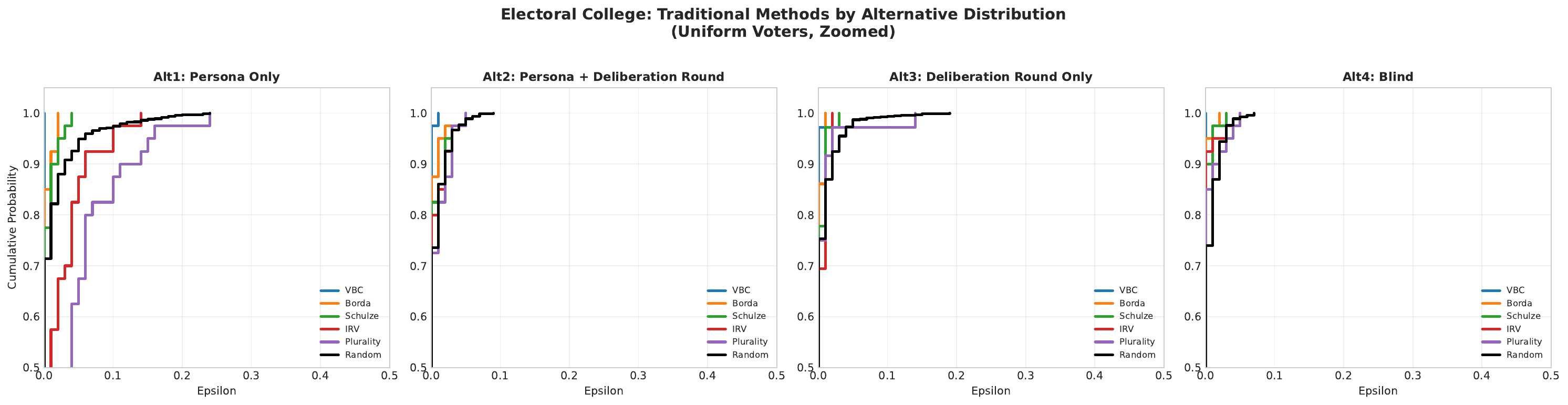}
        \caption{Electoral College}
        \label{fig:cdf_electoral}
    \end{subfigure}

    \vspace{10pt}

    \begin{subfigure}[b]{0.9\textwidth}
        \centering
        \includegraphics[width=\linewidth]{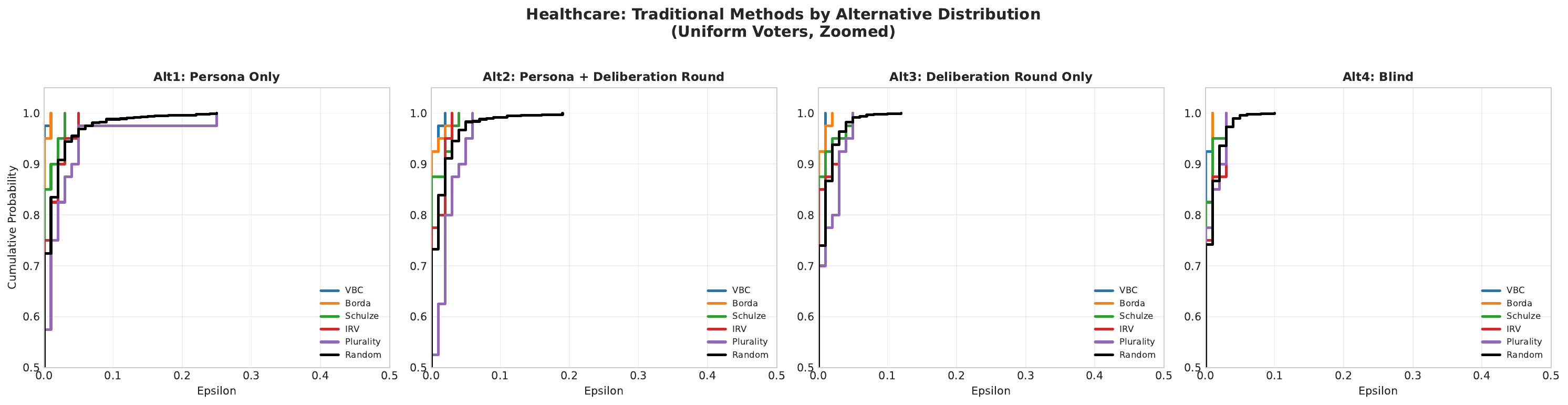}
        \caption{Healthcare}
        \label{fig:cdf_healthcare}
    \end{subfigure}
    \caption{CDF of critical epsilons across various topics (top three topics by degree of polarization). The leftmost subplot (Persona) allows for clearer analysis than other alternative generation methods.} 
    \label{fig:combined_polarization_plots_vertical_top_three}
\end{figure}

\begin{figure}[htbp]
    \centering
    \begin{subfigure}[b]{0.9\textwidth}
        \centering
        \includegraphics[width=\linewidth]{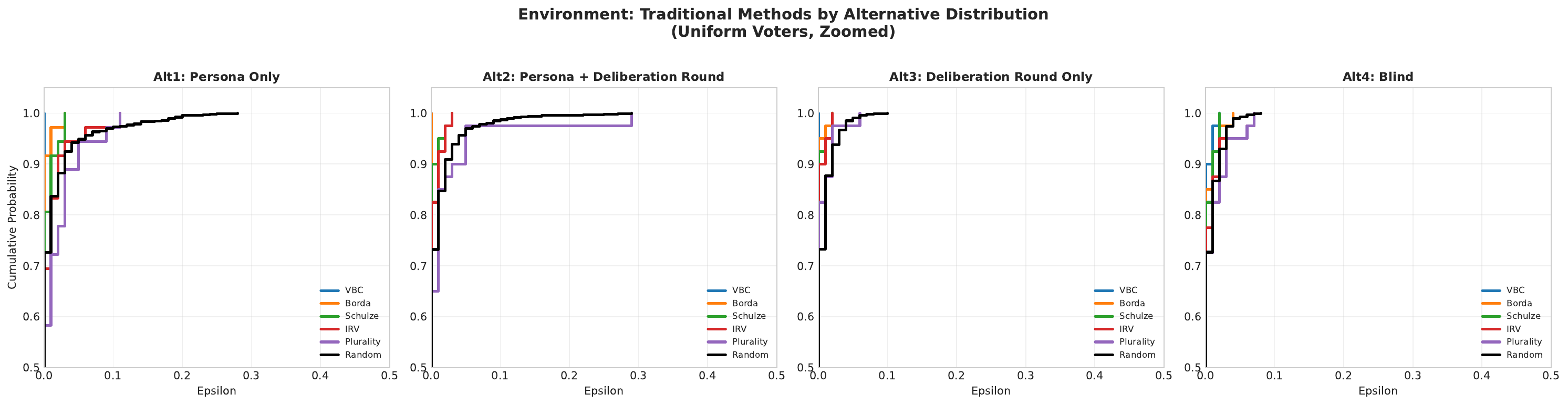}
        \caption{Environment}
        \label{fig:cdf_environment}
    \end{subfigure}

    \vspace{10pt}

    \begin{subfigure}[b]{0.9\textwidth}
        \centering
        \includegraphics[width=\linewidth]{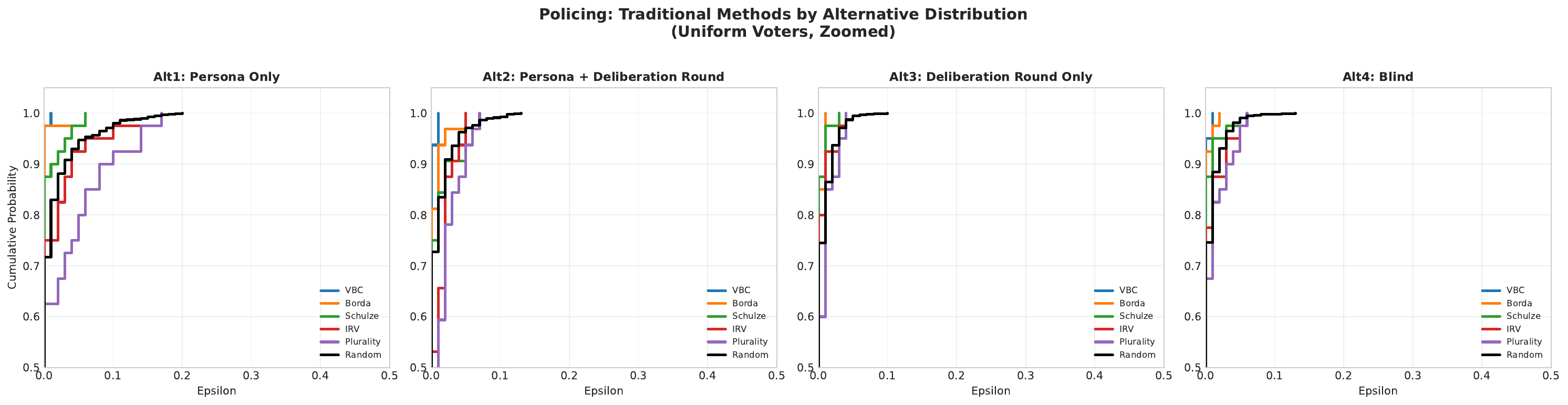}
        \caption{Policing}
        \label{fig:cdf_policing}
    \end{subfigure}

    \vspace{10pt}

    \begin{subfigure}[b]{0.9\textwidth}
        \centering
        \includegraphics[width=\linewidth]{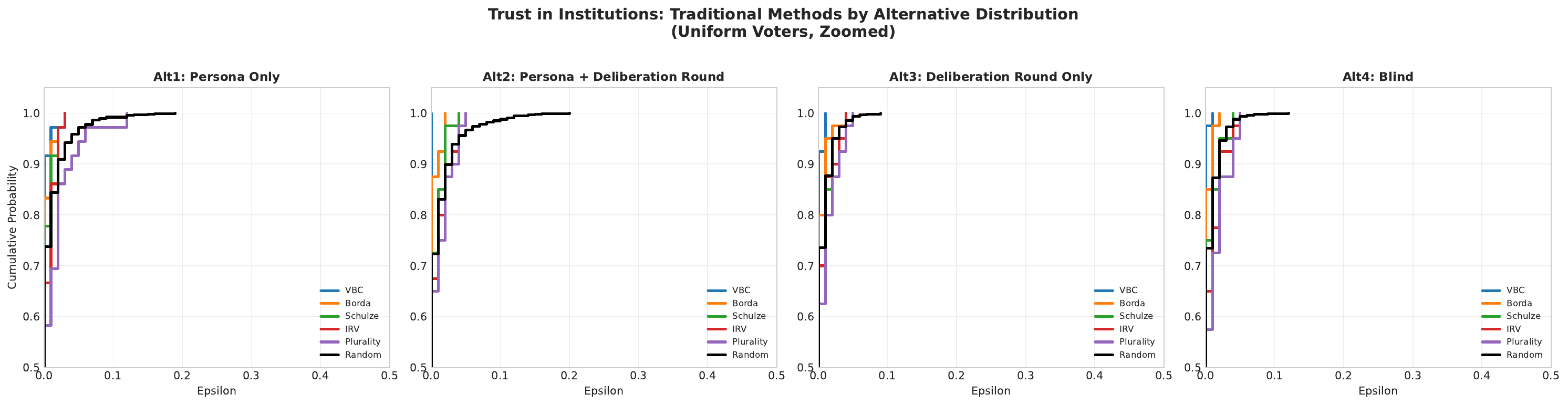}
        \caption{Public Trust}
        \label{fig:cdf_trust}
    \end{subfigure}
    \caption{CDF of critical epsilons across various topics (bottom three topics by degree of polarization). The leftmost subplot (Persona) allows for clearer analysis than other alternative generation methods.}
    \label{fig:combined_polarization_plots_vertical_bottom_three}
\end{figure}

\subsection{Building Preference Profiles}\label{subsec:building-preference-profiles}

In this section, we discuss why we selected iterative ranking as our process for using an LLM to build a preference profile.  

\subsubsection{The Problem of Preference Degeneracy}

Preference degeneracy occurs when the preference rankings produced by generative AI are very similar even for very different profiles. We found that this can happen when the AI is ``overwhelmed'' with the task complexity of ranking too many alternatives at once, and outputs a default ranking of the alternatives (for example, in the order presented to the AI, or its reverse). 

\subsubsection{Mitigating Preference Degeneracy}\label{ssc:mitigate-preference-degeneracy}
There are two natural ways to mitigate preference degeneracy:
\begin{enumerate}
    \item Increase model capabilities.\label{item:capabilities}
    \item Scaffold the task easier to make it easier for the model to complete.\label{item:scaffold}
\end{enumerate}

\Cref{item:capabilities} has the tradeoff of budget concerns. To illustrate, order of magnitude capability jumps between \texttt{gpt-5-nano}, \texttt{gpt-5-mini}, and \texttt{gpt-5.2} are accompanied by similar order of magnitude price jumps. We fix the model as \texttt{gpt-5-mini} as a compromise between capability and price.

Another axis of increasing model capabilities is to increase reasoning effort (the length of the model's chain-of-thought). As the number of tokens get higher, there is also an associated cost increase. Latency also increases as more tokens need to be generated before the tokens related to the extractable ranking is produced. We ran a sweep on the available options for reasoning efforts in \Cref{tab:degeneracy-mitigation}.

There are multiple ways to go about \Cref{item:scaffold}:
\paragraph{Pairwise comparison}\label{par:pairwise-comparison} We repeatedly show the model two alternatives, and ask the model which one to pick, until a full ranking is obtained. However, if there are \(n\) voters and \(m\) alternatives, and an API call takes \(k\) tokens, then the number of tokens required (which is usually dominated by input tokens for seeding techniques) is \(O(mnk\log m)\). For a reasonable experiment of \(m=100,n=100,k=500\), a single preference profile will take about 33 million tokens, which is about \$8.25. 
Using this option, the full experiment will cost: 6 topics \(\times\) 10 repetitions \(\times\) 3 voter distributions \(\times\) \$8.25/profile \(=\) \$1,485.

\paragraph{Direct Ranking}\label{par:direct-ranking} The other extreme is to simply give the model all the statements and ask it to output its ranking in one API call. This is the cheapest possible approach. Unfortunately, this approach frequently gives degenerate rankings, even for the most capable models currently available.

\paragraph{Top-K/Bottom-K Ranking (Iterative Ranking)} Given all remaining alternatives, we ask the model to pick its top K statements and bottom K statements. We then iterate for \(m/2k\) rounds with the subsequent set of remaining alternatives until all alternatives haven been chosen. 

\paragraph{Scoring} Given a set of alternatives, we ask the model to give a score to each alternative. The benefit to this approach is that compared to producing the ranking, the model is not heavily constrained by other alternatives it has processed or will process when processing a new alternative. In other words, the model could ignore the scores it gave to any previous alternatives and focus on giving a score to the current alternative. When the scoring is done, a ranking is induced, and any duplicate scores could be randomly tie-broken. Another approach is to ask the model not to give duplicate scores. 

Because pairwise comparisons were too expensive and direct ranking produced degenerate rankings, we focus on the remaining two methods and report results in \Cref{tab:degeneracy-mitigation}. When we receive a ranking, we screen it for invalid rankings (e.g. hallucinated statement IDs, or duplicate statement IDs, etc). If it is invalid, we retry the query, up to 20 times.

\begin{table}[htbp]
\centering
\caption{Degeneracy Mitigation Experiment Results}
\label{tab:degeneracy-mitigation}
\begin{tabularx}{\textwidth}{>{\raggedright\arraybackslash}X l c c c c}
\toprule
\textbf{Mitigation} & \textbf{Reasoning} & \textbf{Unique} & \textbf{Pres. Order} & \textbf{API} & \textbf{Avg Time} \\
                    & \textbf{Effort}    & \textbf{Rankings} & \textbf{Corr.} & \textbf{Calls} & \textbf{/Call} \\
\midrule
Top-K/Bottom-K & minimal & $79/100$  & $0.029$  & $1307$     & $3.3$s \\
Top-K/Bottom-K & low     & $98/100$  & $0.043$  & $567$      & $11.1$s \\
Top-K/Bottom-K & medium  & $100/100$ & $0.033$  & $\sim 500$ & $\sim 30$s \\
\midrule
Scoring & minimal & $99/100$  & $-0.126$ & $178$      & $8.9$s \\
Scoring & low     & $100/100$ & $-0.987$ & $100$      & $16.0$s \\
Scoring & medium  & $98/100$  & $-0.952$ & $\sim 100$ & $\sim 30$s \\
\bottomrule
\end{tabularx}
\begin{tablenotes}
\small
\item Note: ``Pres. Order Corr.'' stands for presentation order correlations. Values near $\pm 1$ indicate degenerate rankings.
\end{tablenotes}
\end{table}

As shown in \Cref{tab:degeneracy-mitigation}, the scoring method frequently resulted in degenerate rankings as measured by correlation with presentation order. To minimize preference degeneracy, we use the top-K/bottom-K method for all of our experiments. We found that the low reasoning effort produce a high number of unique rankings with minimal retries (67 compared to minimal's 807) and have better latency (11.1s vs medium's 30s). The rankings produced with low reasoning also has a Spearman's correlation of 0.802 with that of medium reasoning. We thus generate the preference profiles for our full experiment using the top-K/bottom-K approach (iterative ranking) using \texttt{gpt-5-mini} on \texttt{low} reasoning effort.

\subsection{Using different methods to find common ground in real election instances}\label{subsec:bridgingness-of-voting-methods-in-real-election-instances}

We validate the pattern in \Cref{eq:vbc-beats-plurality} on real-world election instances from Preflib \citep{MW13}. We plot our results in \Cref{fig:bridgingness-comparison-full}.

\begin{figure}[!htb]
    \centering
    \begin{subfigure}[b]{0.48\textwidth}
        \centering
        \includegraphics[width=\linewidth]{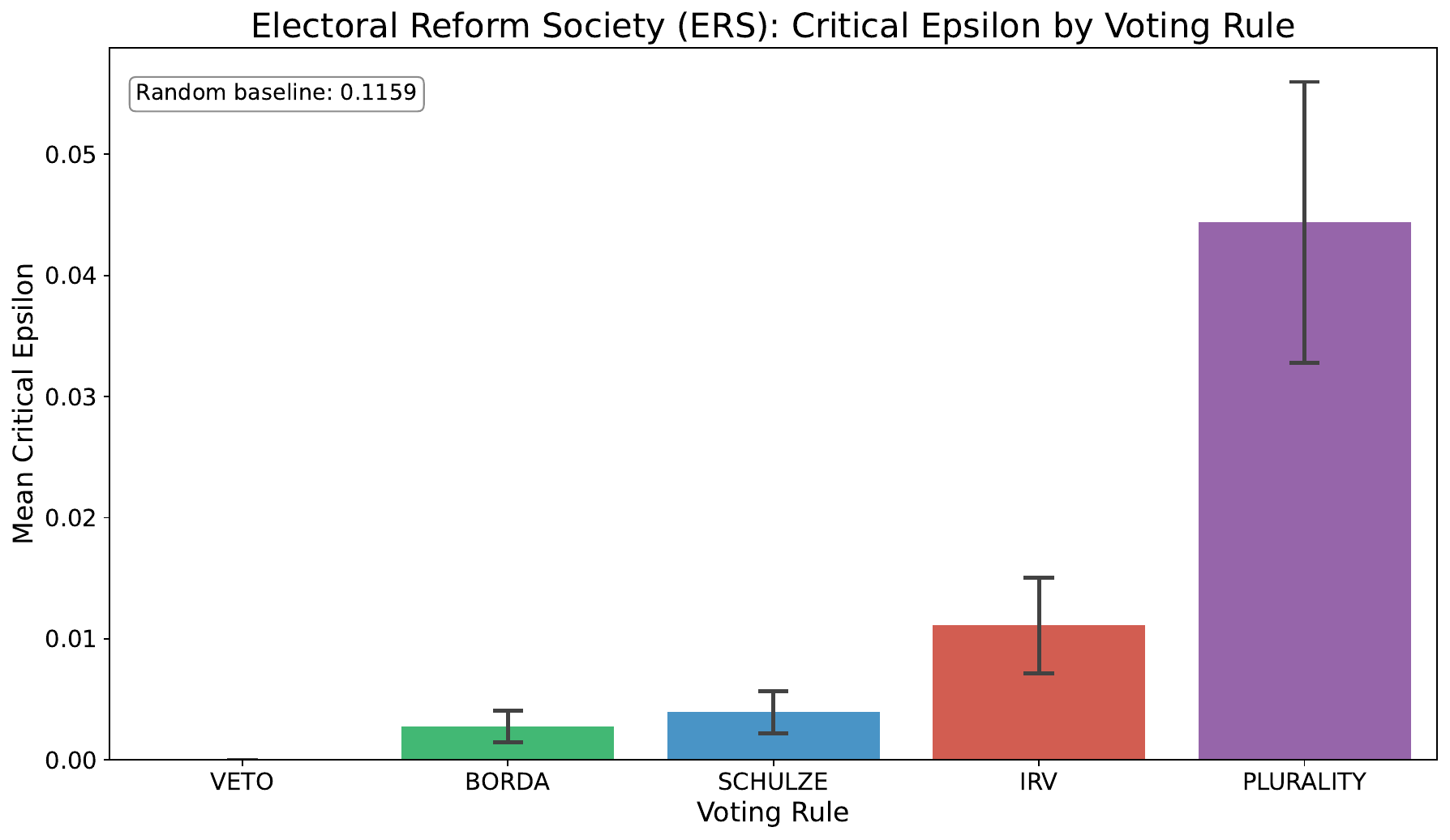}
        \caption{Critical epsilons for Electoral Reform Society \citep{TidemanPlassmann2012,TidemanPlassmann2014a,PlassmannTideman2014}}
        \label{fig:eps-ers}
    \end{subfigure}
    \hfill
    \begin{subfigure}[b]{0.48\textwidth}
        \centering
        \includegraphics[width=\linewidth]{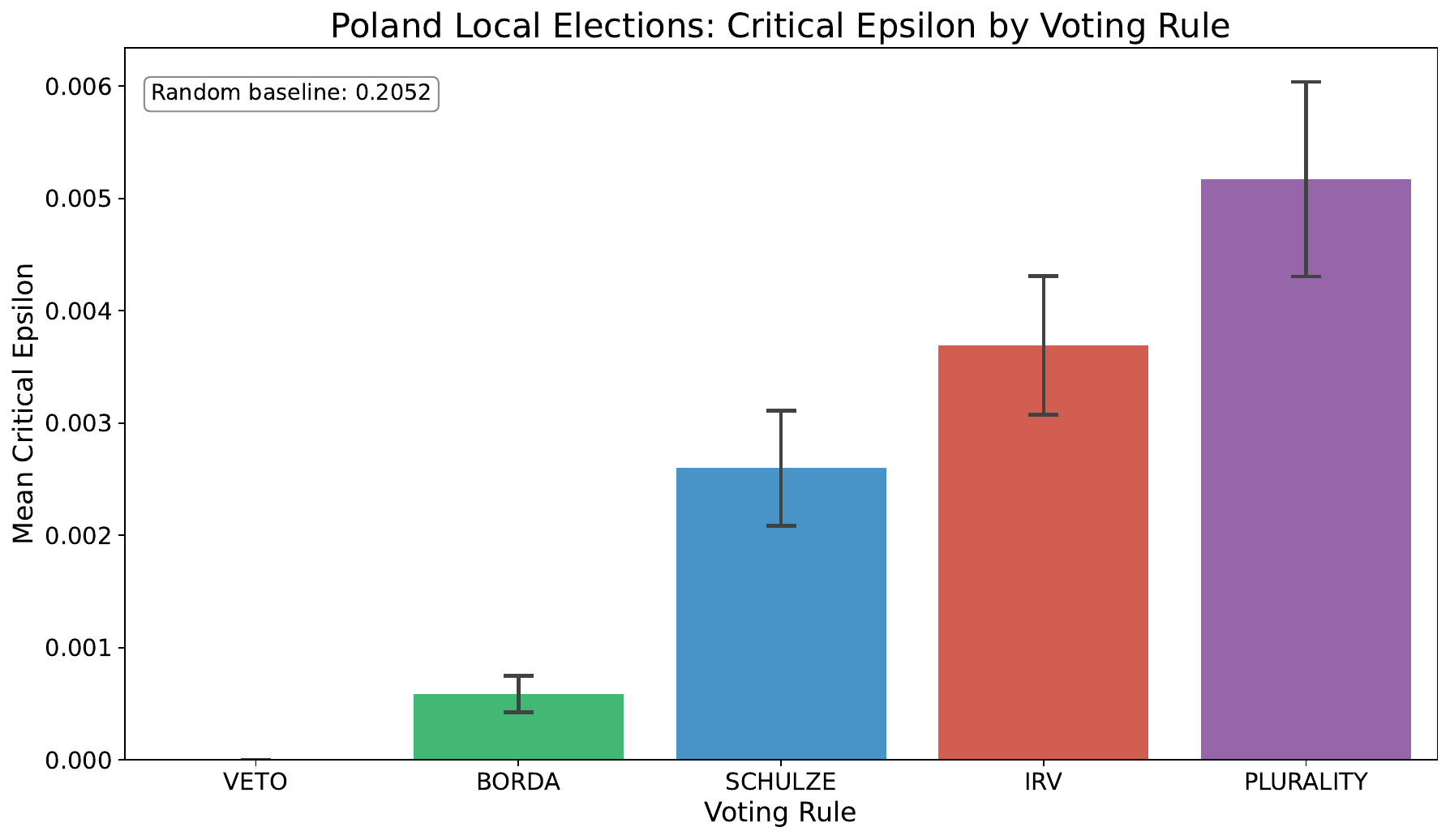}
        \caption{Critical epsilons for Poland Local Elections \citep{Boehmer2022}}
        \label{fig:eps-polish}
    \end{subfigure}

    \vspace{10pt}

    \begin{subfigure}[b]{0.48\textwidth}
        \centering
        \includegraphics[width=\linewidth]{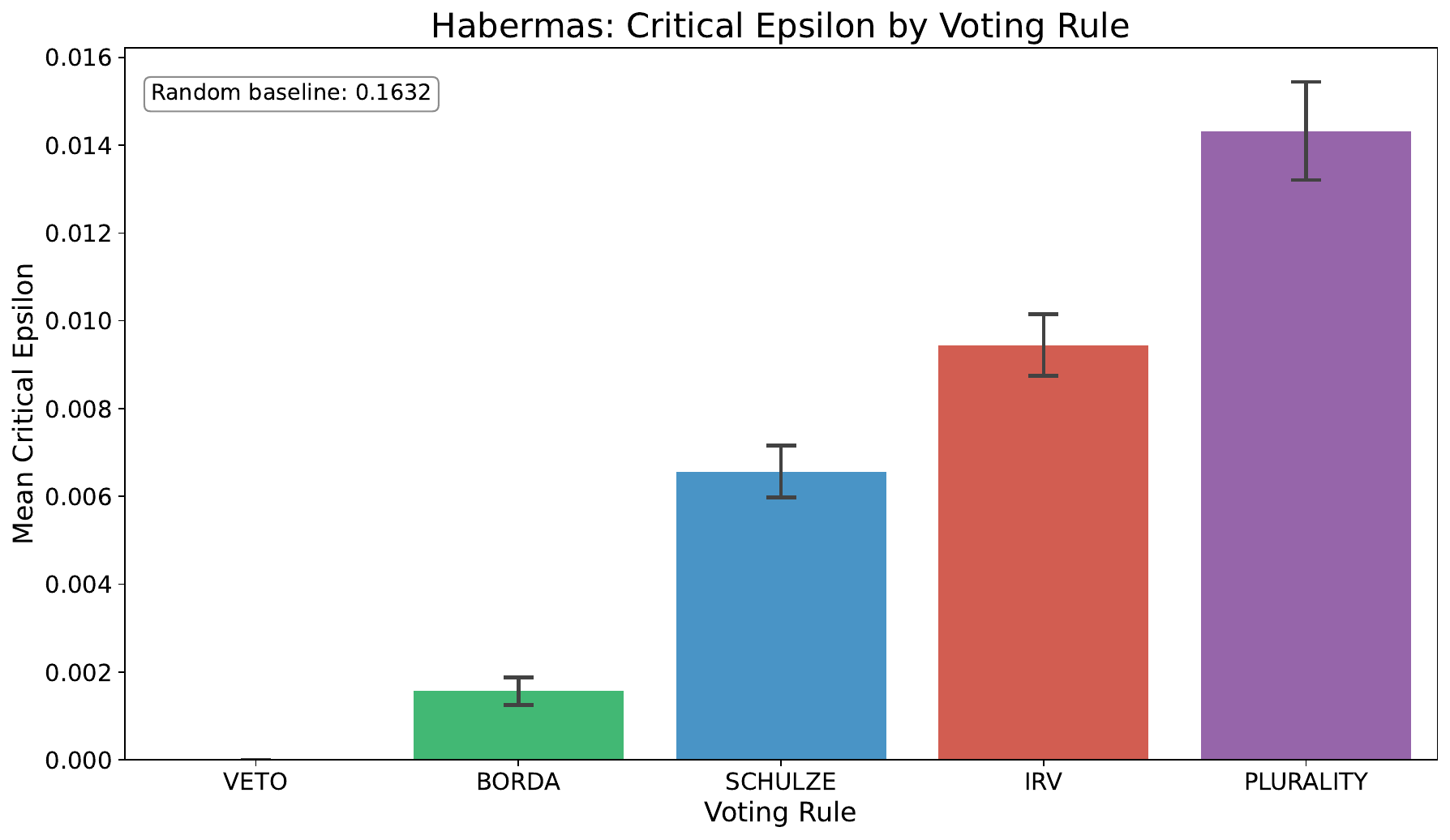}
        \caption{Critical epsilons for Habermas \citep{TBJ+24}}
        \label{fig:eps-habermas}
    \end{subfigure}
    \hfill
    \begin{subfigure}[b]{0.48\textwidth}
        \centering
        \includegraphics[width=\linewidth]{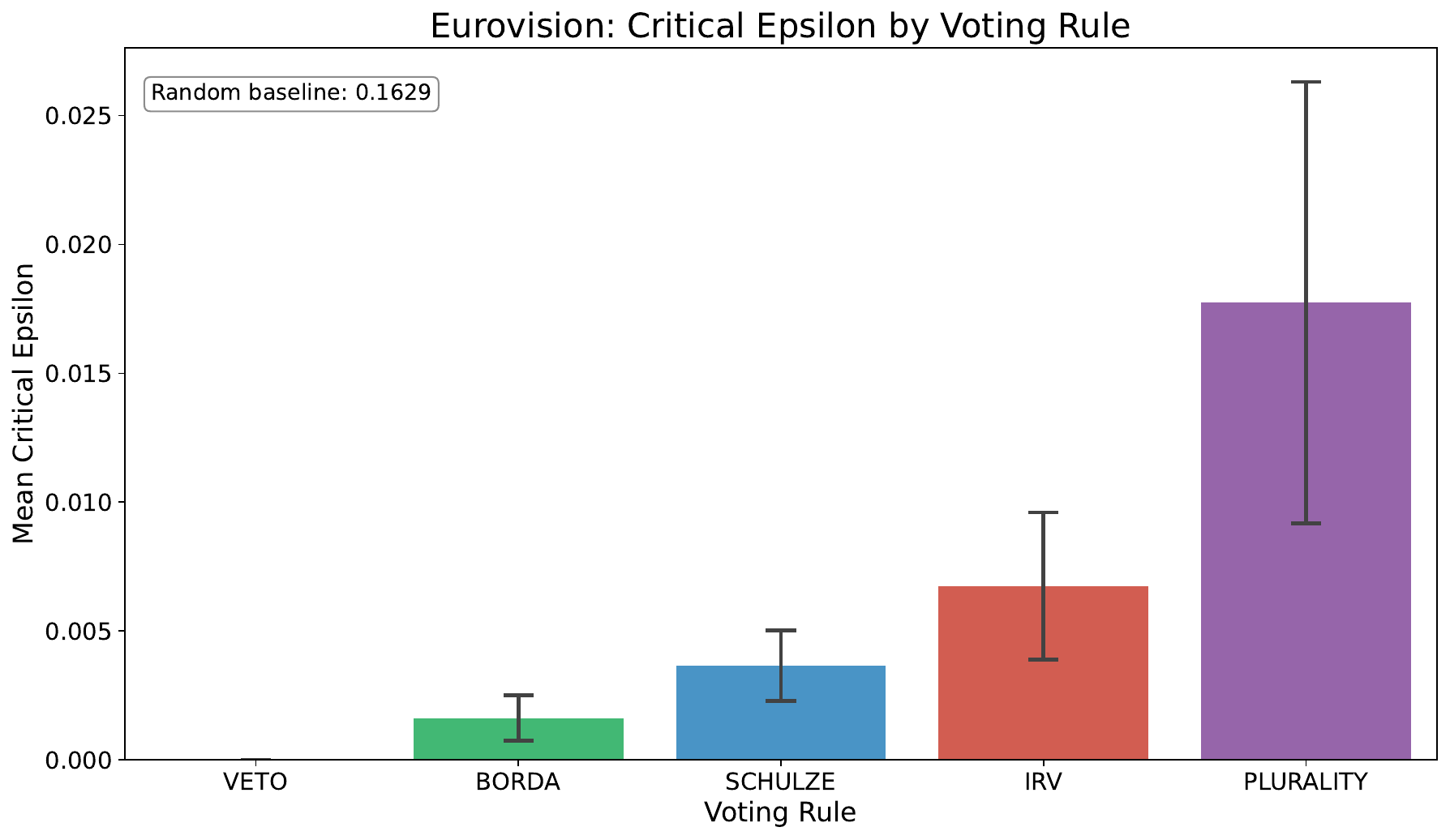}
        \caption{Critical epsilons for Eurovision \citep{Boehmer2023}}
        \label{fig:eps-eurovision}
    \end{subfigure}

    \caption{Comparison of critical epsilons across voting rules on various real-world datasets from PrefLib. Standard errors depicted.}
    \label{fig:bridgingness-comparison-full}
\end{figure}

\subsection{Inserting a New Statement to a Preference Profile}\label{subsec:inserting-a-new-statement-to-a-preference-profile}
In this section, we discuss different methods for inserting a single alternative into an already constructed ranking of 100 existing alternatives, which is a necessary step for evaluating generative voting rules.
\begin{itemize}
    \item Insert into 100: we give the model all 100 statements and ask the model to output the position for the new statement.
    \item Chunked insertion: we split the 100 statements into 5 contiguous chunks. We ask the model to output the position it will insert the new statement into for each chunk. We then tally the positions and according to its chunk assign the appropriate ranking. For example, if the new statement gets position 0, 3, 15, 8, 9 in the 5 chunks, then the insertion position is \(0+3+15+8+9=35\).
    \item Pairwise comparison: we ask the model to compare the new statement with every other statement in separate API calls to obtain the insertion position.
    \item Iterative ranking: we rerun the algorithm for building preference profiles as in \Cref{para:preferences-profile} and extract the resulting position of the new statement as its insertion position.
\end{itemize}
In \Cref{fig:insertion-methods-comparison}, we evaluate the performance of these insertion algorithms by randomly selecting a statement from a preference profile to be (re-)inserted. We did this for 10 randomly sampled statements and 100 voters for 1,000 insertion tests for each method. We found that the iterative ranking provided the best fit while still being cost-efficient, and this is the method used in our main experiments.

\begin{figure}[!htb]
    \centering
    \begin{subfigure}[b]{0.48\textwidth}
        \centering
        \includegraphics[width=\linewidth]{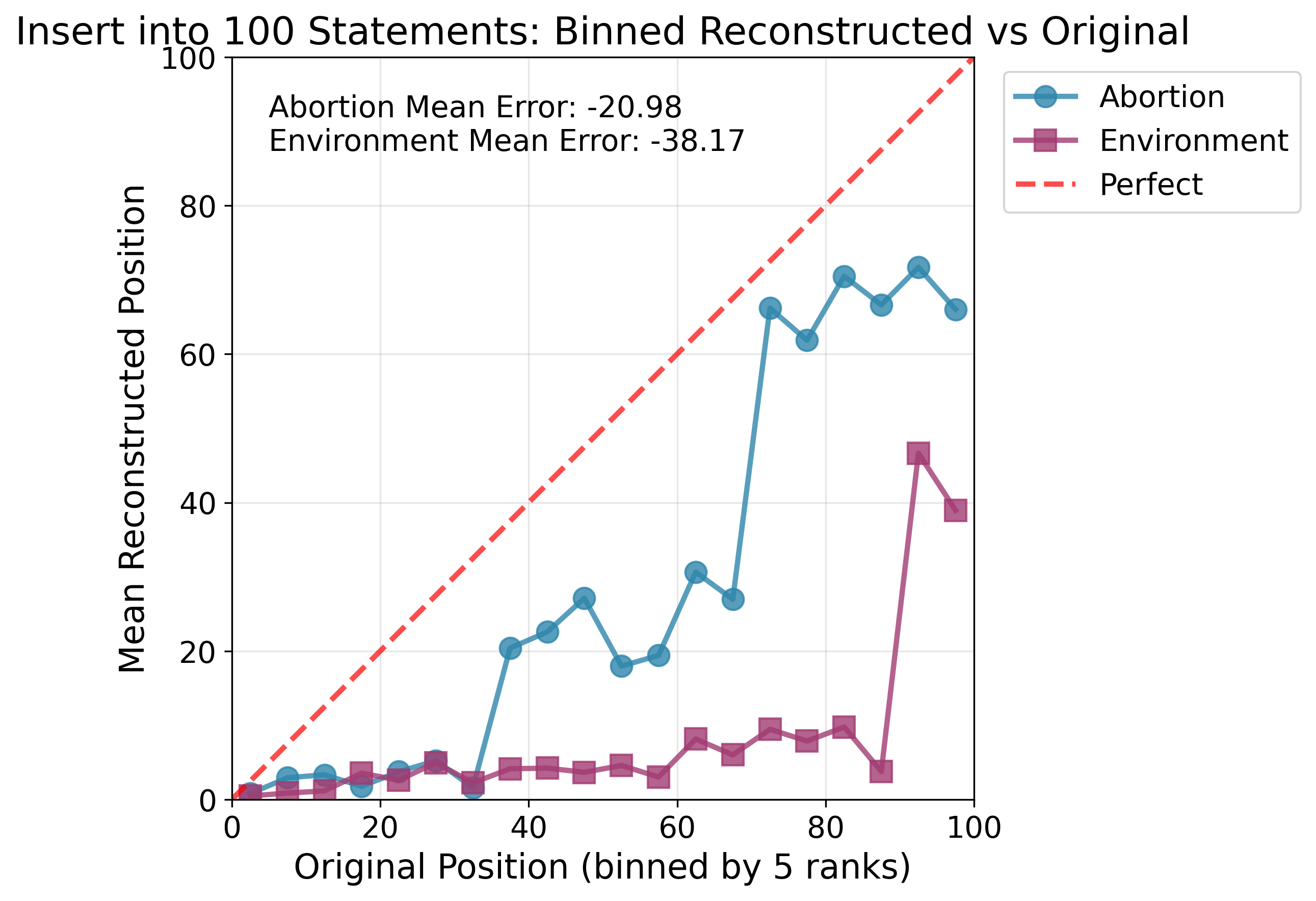}
        \caption{Insert 100 Method}
        \label{fig:insert_100}
    \end{subfigure}
    \hfill
    \begin{subfigure}[b]{0.48\textwidth}
        \centering
        \includegraphics[width=\linewidth]{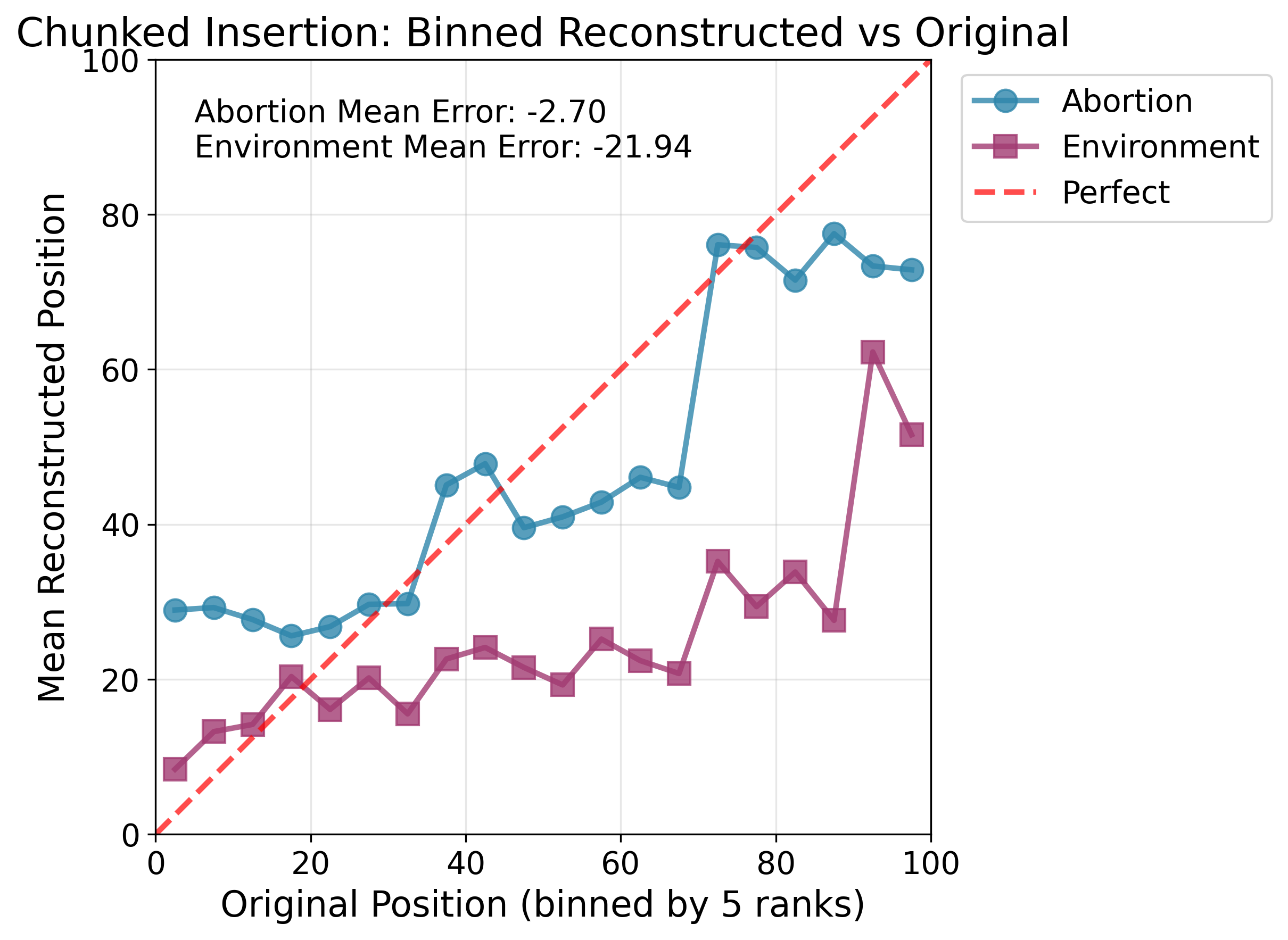}
        \caption{Chunked Method}
        \label{fig:chunked}
    \end{subfigure}

    \vspace{10pt} 
    
    \begin{subfigure}[b]{0.48\textwidth}
        \centering
        \includegraphics[width=\linewidth]{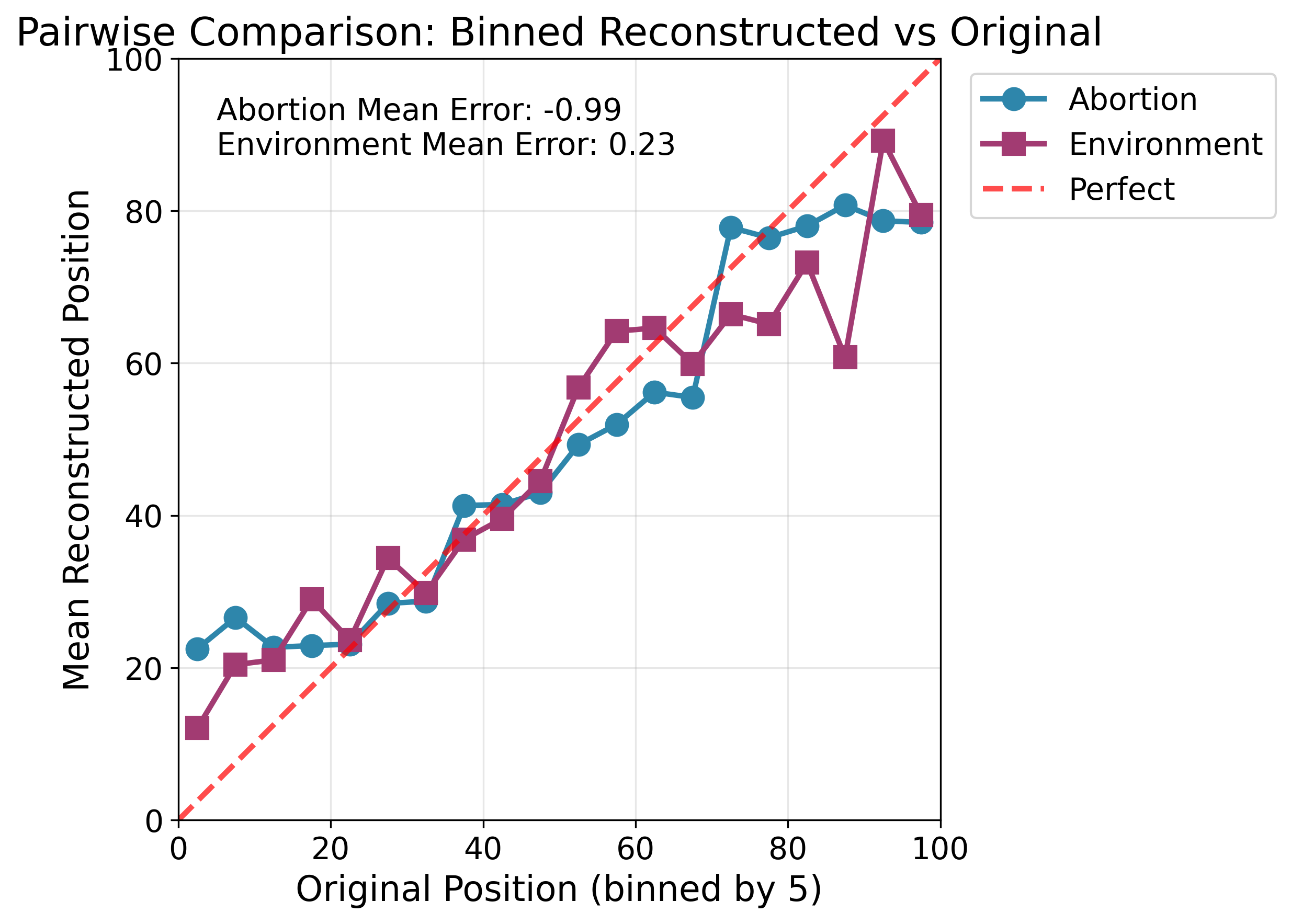}
        \caption{Pairwise Method}
        \label{fig:pairwise}
    \end{subfigure}
    \hfill
    \begin{subfigure}[b]{0.48\textwidth}
        \centering
        \includegraphics[width=\linewidth]{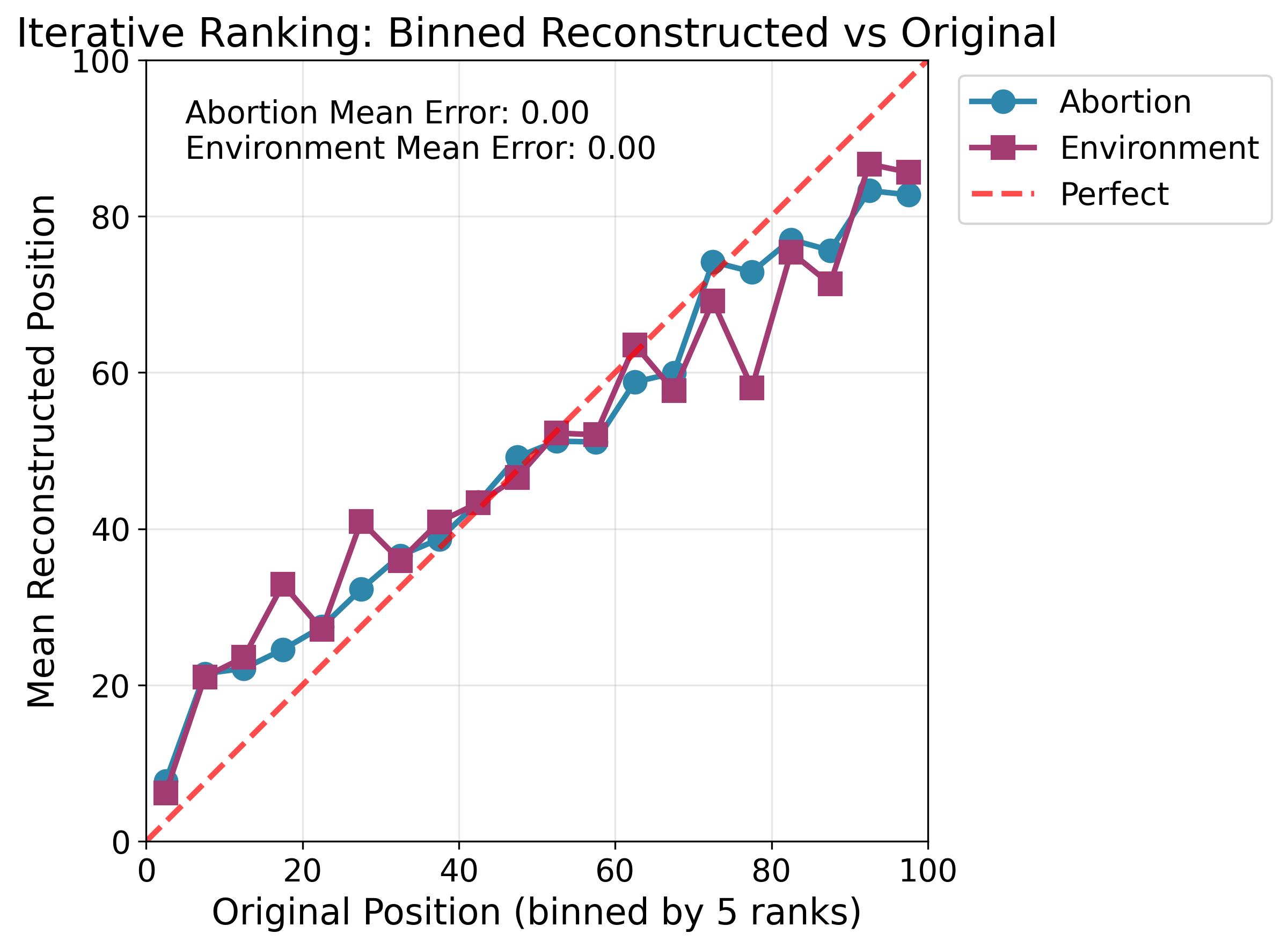}
        \caption{Iterative Ranking Method}
        \label{fig:iterative}
    \end{subfigure}
    
    \caption{Binned scatter plots comparing various statement insertion methods by plotting reconstruction position versus original position. Error units are in the ranks, e.g. an error of -20 means that the method ranks the new statement on average 20 positions below its original.}
    \label{fig:insertion-methods-comparison}
\end{figure}

\subsection{Detailed Results for Critical Epsilons of Voting Rules}\label{subsec:detailed-results-epsilons}
 In \Cref{tab:p99-epsilon,tab:pl0.01-epsilon,fig:cdf_traditional_group2}, we present additional detailed statistics of the critical epsilons achieved by various traditional voting rules.

\begin{table}[!htb]
\centering
\caption{P99 of critical epsilons for various voting rules across various topics}\label{tab:p99-epsilon}
\resizebox{\columnwidth}{!}{
\begin{tabular}{lcccccc}
\toprule
Method & Abortion & Electoral College & Healthcare & Policing & Environment & Trust in Institutions \\
\midrule
VBC & \textbf{0.0161} & \textbf{0.0000} & \textbf{0.0061} & \textbf{0.0061} & \textbf{0.0000} & \textbf{0.0230} \\
Borda & 0.1432 & 0.0200 & 0.0100 & 0.0366 & 0.0230 & 0.0265 \\
Schulze & 0.1432 & 0.0361 & 0.0300 & 0.0522 & 0.0300 & 0.0265 \\
IRV & 0.2327 & 0.1244 & 0.0500 & 0.1427 & 0.0925 & 0.0265 \\
Random & 0.3101 & 0.1700 & 0.1200 & 0.1401 & 0.1801 & 0.0801 \\
Plurality & 0.3344 & 0.2088 & 0.1720 & 0.1583 & 0.1030 & 0.0990 \\
\bottomrule
\end{tabular}

}
\end{table}

\begin{table}[!htb]
\centering
\caption{Percentage of critical epsilons that are less than 0.01 for various voting rules across various topics}\label{tab:pl0.01-epsilon}
\resizebox{\columnwidth}{!}{
\begin{tabular}{lcccccc}
\toprule
Method & Abortion & Electoral College & Healthcare & Policing & Environment & Trust in Institutions \\
\midrule
VBC & \textbf{95.0000} & \textbf{100.0000} & \textbf{97.5000} & \textbf{97.5000} & \textbf{100.0000} & \textbf{91.6667} \\
Borda & 87.5000 & 85.0000 & 95.0000 & \textbf{97.5000} & 91.6667 & 83.3333 \\
Schulze & 75.0000 & 77.5000 & 85.0000 & 87.5000 & 80.5556 & 77.7778 \\
IRV & 65.0000 & 42.5000 & 75.0000 & 75.0000 & 69.4444 & 66.6667 \\
Random & 58.1000 & 71.4000 & 72.4000 & 71.7000 & 72.6667 & 73.7778 \\
Plurality & 45.0000 & 22.5000 & 57.5000 & 62.5000 & 58.3333 & 58.3333 \\
\bottomrule
\end{tabular}
}
\end{table}

\begin{figure}[!htb]
    \centering
    \includegraphics[width=\linewidth]{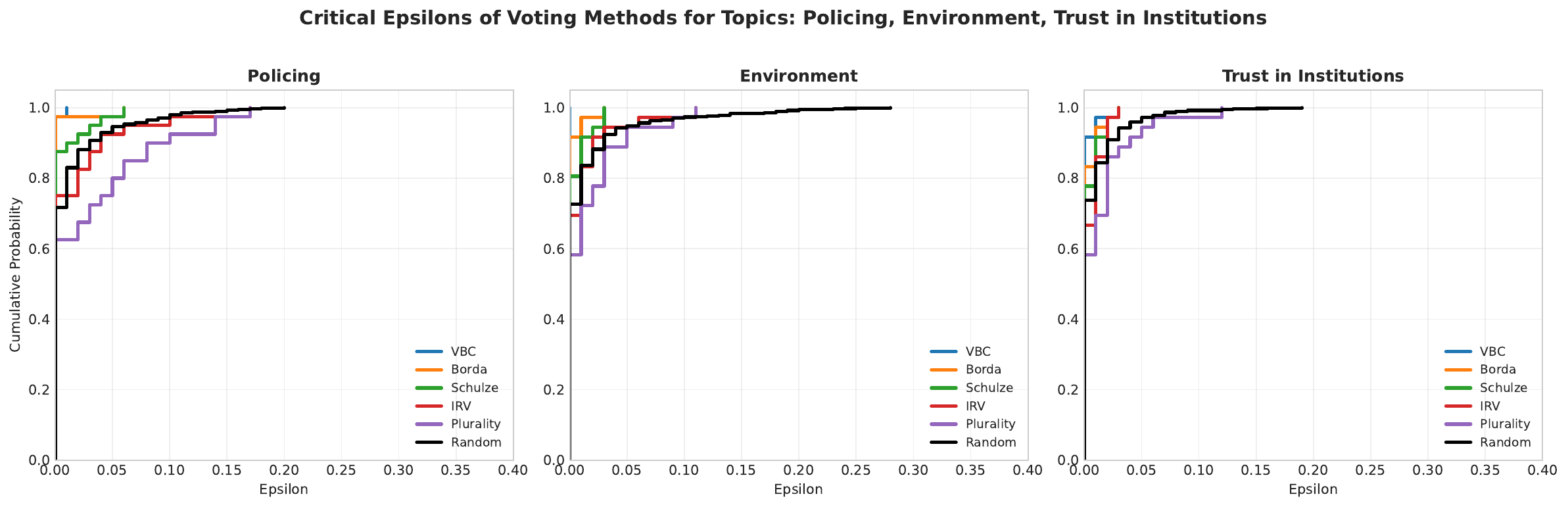}
    \caption{CDF of critical epsilons for least controversial topics for various voting rules.}
    \label{fig:cdf_traditional_group2}
\end{figure}

\FloatBarrier
\subsection{LLM-Powered Voting Rules as a Selector}\label{subsec:selector}
In addition to testing LLM methods that generate their own statement, we also tested methods that use LLMs to select an existing statement. We tested six different variants that vary in how much information is given to the model.
\begin{itemize}
    \item GPT-Select: the model is given the 20 sampled statements.
    \item GPT-Select+Rankings: the model is given the 20 sampled statements, and the $20\times 20$ preference profile.
    \item GPT-Select+Personas: the model is given the 20 sampled statements and the personas of the 20 voters.
    \item GPT-Full: the model is given the 100 statements.
    \item GPT-Full+Rankings: the model is given the 100 statements, and the $20\times 100$ preference profile.
    \item GPT-Full+Personas: the model is given the 100 statements and the personas of the 20 voters.
\end{itemize}
We present the results in \Cref{tab:crit-eps-selector-methods,tab:p99-critical-epsilon-selection-methods,tab:pl0.01-critical-epsilon-selection-methods}. VBC still outperforms the LLM methods. Interestingly, that providing more information, like voter personas or the rankings, slightly degrades performance.

\begin{table}[!htb]
\centering
\caption{Mean critical epsilons for VBC, voting rules with GPT as a selector and a random baseline across various topics. Smallest epsilons are bolded.}
\label{tab:crit-eps-selector-methods}
\resizebox{\columnwidth}{!}{
\begin{tabular}{lcccccc}
\toprule
Method & Abortion & Electoral College & Healthcare & Policing & Environment & Trust in Institutions \\
\midrule
VBC & \textbf{0.0008} & \textbf{0.0000} & \textbf{0.0003} & \textbf{0.0003} & \textbf{0.0000} & \textbf{0.0014} \\
GPT-Select & 0.0038 & 0.0048 & 0.0048 & 0.0108 & 0.0103 & 0.0094 \\
GPT-Sel+Rank & 0.0143 & 0.0365 & 0.0163 & 0.0085 & 0.0075 & 0.0056 \\
GPT-Sel+Pers & 0.0125 & 0.0028 & 0.0095 & 0.0195 & 0.0100 & 0.0067 \\
GPT-Full & 0.0013 & 0.0050 & 0.0068 & 0.0055 & 0.0014 & 0.0053 \\
GPT-Full+Rank & 0.0593 & 0.0425 & 0.0055 & 0.0115 & 0.0114 & 0.0075 \\
GPT-Full+Pers & 0.0050 & 0.0023 & 0.0083 & 0.0143 & 0.0086 & 0.0039 \\
Random & 0.0609 & 0.0108 & 0.0082 & 0.0103 & 0.0108 & 0.0072 \\
\bottomrule
\end{tabular}

}
\end{table}

\begin{table}[!htb]
\centering
\caption{P99 critical epsilons for for VBC, voting rules with GPT as a selector and a random baseline across various topics. Smallest epsilon for each topic is bolded.}
\label{tab:p99-critical-epsilon-selection-methods}
\resizebox{\columnwidth}{!}{
\begin{tabular}{lcccccc}
\toprule
Method & Abortion & Electoral College & Healthcare & Policing & Environment & Trust in Institutions \\
\midrule
VBC & 0.0161 & \textbf{0.0000} & \textbf{0.0061} & \textbf{0.0061} & \textbf{0.0000} & \textbf{0.0230} \\
GPT-Select & 0.0505 & 0.0422 & 0.0361 & 0.1166 & 0.1190 & 0.0400 \\
GPT-Sel+Rank & 0.1859 & 0.2088 & 0.2149 & 0.0722 & 0.0820 & 0.0265 \\
GPT-Sel+Pers & 0.1954 & 0.0422 & 0.0900 & 0.0900 & 0.0865 & 0.0365 \\
GPT-Full & \textbf{0.0100} & 0.0322 & 0.0666 & 0.0805 & 0.0100 & 0.0400 \\
GPT-Full+Rank & 0.3461 & 0.2205 & 0.0461 & 0.0922 & 0.0925 & 0.0465 \\
GPT-Full+Pers & 0.1032 & 0.0261 & 0.1254 & 0.1000 & 0.1700 & 0.0465 \\
Random Insertion & 0.3161 & 0.0861 & 0.0622 & 0.1449 & 0.0300 & 0.0600 \\
Random & 0.3101 & 0.1700 & 0.1200 & 0.1401 & 0.1801 & 0.0801 \\
\bottomrule
\end{tabular}

}
\end{table}

\begin{table}[!htb]
\centering
\caption{Percentage of critical epsilons that are less than 0.01 for VBC, voting rules with GPT as a selector and a random baseline across various topics. Largest percentages for each topic is bolded.}
\label{tab:pl0.01-critical-epsilon-selection-methods}
\resizebox{\columnwidth}{!}{%
\begin{tabular}{lcccccc}
\toprule
Method & Abortion & Electoral College & Healthcare & Policing & Environment & Trust in Institutions \\
\midrule
VBC & \textbf{0.950} & \textbf{1.000} & \textbf{0.975} & \textbf{0.975} & \textbf{1.000} & \textbf{0.917} \\
GPT-Select & 0.850 & 0.775 & 0.750 & 0.725 & 0.750 & 0.667 \\
GPT-Sel+Rank & 0.700 & 0.450 & 0.700 & 0.825 & 0.694 & 0.639 \\
GPT-Sel+Pers & 0.900 & 0.900 & 0.775 & 0.550 & 0.778 & 0.722 \\
GPT-Full & 0.875 & 0.700 & 0.675 & 0.875 & 0.861 & 0.778 \\
GPT-Full+Rank & 0.500 & 0.350 & 0.775 & 0.700 & 0.611 & 0.583 \\
GPT-Full+Pers & 0.900 & 0.875 & 0.750 & 0.650 & 0.889 & 0.889 \\
Random Insertion & 0.500 & 0.650 & 0.600 & 0.725 & 0.550 & 0.600 \\
Random & 0.581 & 0.714 & 0.724 & 0.717 & 0.727 & 0.738 \\
\bottomrule
\end{tabular}

}
\end{table}

\FloatBarrier
\subsection{Details on Clustered Voters}
\label{subsec:details-on-clustered-voters}

\subsubsection{Keywords Used for Classification}
To cluster voters, we applied case-insensitive keyword matching on the \verb|idealogy| field in the personas using the following sets of keywords.

\paragraph{Conservative/Traditional}

\begin{itemize}
    \item conservative
    \item traditional
    \item traditionalist
    \item libertarian
    \item libertarianism
    \item fiscal conservatism
    \item social conservatism
    \item christian values
    \item family values
    \item small government
\end{itemize}

\paragraph{Progressive/Liberal}

\begin{itemize}
    \item progressive
    \item liberal
    \item social liberal
    \item social justice
    \item feminist
    \item feminism
    \item egalitarian
    \item environmentalism
    \item environmentalist
    \item environmental conservation
    \item environmental protection
    \item social equality
    \item workers' rights
    \item social welfare
    \item humanism
    \item humanistic
    \item humanitarian
\end{itemize}

\subsubsection{More Results on Conservative Voters}
We report additional statistics in \Cref{tab:p99-critical-epsilon-generative-methods-conservative,tab:pl0.01-critical-epsilon-generative-methods-conservative}.

\begin{table}[!htb]
\centering
\caption{P99 critical epsilons for for VBC, generative voting rules and random baselines across various topics for conservative voters. Smallest epsilon for each topic is bolded.}
\label{tab:p99-critical-epsilon-generative-methods-conservative}
\resizebox{\columnwidth}{!}{
\begin{tabular}{lcccccc}
\toprule
Method & Abortion & Electoral College & Healthcare & Policing & Environment & Trust in Institutions \\
\midrule
VBC & \textbf{0.0161} & \textbf{0.0061} & \textbf{0.0061} & \textbf{0.0000} & \textbf{0.0061} & \textbf{0.0000} \\
GPT-Blind & 0.5005 & 0.0900 & 0.0461 & 0.1022 & 0.3588 & 0.1161 \\
GPT-Synthesize & 0.3300 & 0.1832 & 0.0944 & 0.2298 & 0.1166 & 0.3461 \\
GPT-Synth+Rank. & 0.3544 & 0.1261 & 0.1144 & 0.1022 & 0.2722 & 0.0300 \\
GPT-Synth.+Pers. & 0.2700 & 0.0461 & 0.0322 & 0.0200 & 0.0361 & 0.0100 \\
Random Insertion & 0.3261 & 0.4805 & 0.3683 & 0.3600 & 0.1000 & 0.2286 \\
Random & 0.4601 & 0.5100 & 0.4200 & 0.4201 & 0.3301 & 0.3901 \\
\bottomrule
\end{tabular}
}
\end{table}

\begin{table}[!htb]
\centering
\caption{Percentage of critical epsilons that are less than 0.01 for VBC, generative voting rules and random baselines across various topics for conservative voters. Largest percentages for each topic is bolded.}
\label{tab:pl0.01-critical-epsilon-generative-methods-conservative}
\resizebox{\columnwidth}{!}{
\begin{tabular}{lcccccc}
\toprule
Method & Abortion & Electoral College & Healthcare & Policing & Environment & Trust in Institutions \\
\midrule
VBC & \textbf{0.900} & \textbf{0.975} & \textbf{0.975} & \textbf{1.000} & \textbf{0.975} & \textbf{1.000} \\
GPT-Blind & 0.000 & 0.225 & 0.750 & 0.425 & 0.100 & 0.275 \\
GPT-Synthesize & 0.000 & 0.350 & 0.325 & 0.175 & 0.175 & 0.325 \\
GPT-Synth+Rank. & 0.025 & 0.400 & 0.575 & 0.425 & 0.500 & 0.675 \\
GPT-Synth.+Pers. & 0.250 & 0.550 & 0.725 & 0.750 & 0.575 & 0.900 \\
Random Insertion & 0.425 & 0.100 & 0.575 & 0.400 & 0.525 & 0.750 \\
Random & 0.338 & 0.242 & 0.609 & 0.389 & 0.484 & 0.539 \\
\bottomrule
\end{tabular}

}
\end{table}

\subsubsection{Results on Progressive Voters}
\label{app:progressive}
We present the critical epsilons of generative voting rules on progressive voters in \Cref{tab:mean-critical-epsilon-generative-methods-progressive,tab:p99-critical-epsilon-generative-methods-progressive,tab:pl0.01-critical-epsilon-generative-methods-progressive}.

\begin{table}[!htb]
\centering
\caption{Mean critical epsilons for for VBC, generative voting rules and random baselines across various topics for progressive voters. Smallest epsilon for each topic is bolded.}
\label{tab:mean-critical-epsilon-generative-methods-progressive}
\resizebox{\columnwidth}{!}{
\begin{tabular}{lcccccc}
\toprule
Method & Abortion & Electoral College & Healthcare & Policing & Environment & Trust in Institutions \\
\midrule
VBC & \textbf{0.0000} & \textbf{0.0018} & \textbf{0.0000} & \textbf{0.0005} & \textbf{0.0000} & \textbf{0.0003} \\
GPT-Blind & 0.0185 & 0.0280 & 0.0113 & 0.0048 & 0.0035 & 0.0215 \\
GPT-Synthesize & 0.0095 & 0.0103 & 0.0058 & 0.0135 & 0.0108 & 0.0200 \\
GPT-Synth+Rank. & 0.0055 & 0.0075 & 0.0005 & 0.0025 & 0.0025 & 0.0093 \\
GPT-Synth.+Pers. & 0.0058 & 0.0050 & 0.0013 & 0.0033 & 0.0018 & 0.0143 \\
Random Insertion & 0.2205 & 0.0130 & 0.0220 & 0.0532 & 0.0388 & 0.0120 \\
Random & 0.1751 & 0.0989 & 0.0428 & 0.0625 & 0.0382 & 0.0330 \\
\bottomrule
\end{tabular}

}
\end{table}

\begin{table}[!htb]
\centering
\caption{P99 critical epsilons for for VBC, generative voting rules and random baselines across various topics for progressive voters. Smallest epsilon for each topic is bolded.}
\label{tab:p99-critical-epsilon-generative-methods-progressive}
\resizebox{\columnwidth}{!}{
\begin{tabular}{lcccccc}
\toprule
Method & Abortion & Electoral College & Healthcare & Policing & Environment & Trust in Institutions \\
\midrule
VBC & \textbf{0.0000} & 0.0405 & \textbf{0.0000} & \textbf{0.0122} & \textbf{0.0000} & \textbf{0.0069} \\
GPT-Blind & 0.1622 & 0.1805 & 0.0822 & 0.0261 & 0.0283 & 0.0922 \\
GPT-Synthesize & 0.0844 & 0.0600 & 0.0300 & 0.1710 & 0.0927 & 0.1161 \\
GPT-Synth+Rank. & 0.0200 & 0.0461 & 0.0100 & 0.0322 & 0.0161 & 0.0661 \\
GPT-Synth.+Pers. & 0.0200 & \textbf{0.0300} & 0.0161 & 0.0300 & 0.0100 & 0.1083 \\
Random Insertion & 0.6444 & 0.1781 & 0.2544 & 0.3966 & 0.3149 & 0.0761 \\
Random & 0.4601 & 0.5100 & 0.4200 & 0.4201 & 0.3301 & 0.3901 \\
\bottomrule
\end{tabular}

}
\end{table}

\begin{table}[!htb]
\centering
\caption{Percentage of critical epsilons that are less than 0.01 for VBC, generative voting rules and random baselines across various topics for progressive voters. Largest percentages for each topic is bolded.}
\label{tab:pl0.01-critical-epsilon-generative-methods-progressive}
\resizebox{\columnwidth}{!}{
\begin{tabular}{lcccccc}
\toprule
Method & Abortion & Electoral College & Healthcare & Policing & Environment & Trust in Institutions \\
\midrule
VBC & \textbf{1.000} & \textbf{0.950} & \textbf{1.000} & \textbf{0.975} & \textbf{1.000} & \textbf{0.969} \\
GPT-Blind & 0.500 & 0.350 & 0.600 & 0.650 & 0.725 & 0.350 \\
GPT-Synthesize & 0.600 & 0.500 & 0.575 & 0.625 & 0.525 & 0.350 \\
GPT-Synth+Rank. & 0.650 & 0.600 & 0.950 & 0.850 & 0.775 & 0.650 \\
GPT-Synth.+Pers. & 0.650 & 0.650 & 0.900 & 0.800 & 0.825 & 0.500 \\
Random Insertion & 0.275 & 0.600 & 0.675 & 0.625 & 0.550 & 0.600 \\
Random & 0.338 & 0.242 & 0.609 & 0.389 & 0.484 & 0.539 \\
\bottomrule
\end{tabular}

}
\end{table}

\FloatBarrier
\subsection{Prompts}\label{subsec:prompts}
Here we present the prompt templates used to generate statements, build preference profiles, and insert statements.

\subsubsection{Generating Alternatives}

\paragraph{With persona, without deliberation round}
These are the prompts we used for the main experiments.

\begin{promptbox}{System Prompt}
You are writing a statement that reflects your perspective on a topic.
\end{promptbox}

\begin{promptbox}{User Prompt}
You are a person with the following characteristics:
{persona}

Topic: "{topic}"

Write a bridging statement expressing your views on this topic. Your statement should:
- Reflect your background, values, and life experiences
- Aim to find common ground or bridge different viewpoints
- Be 2-4 sentences long
- NOT write in first-person
- NOT explicitly reference your identity or demographics (avoid "As a [X]...")

Write only the statement:
\end{promptbox}

\paragraph{With persona, with deliberation round}
\begin{promptbox}{System Prompt}
You are writing a statement that reflects your perspective on a topic.
\end{promptbox}

\begin{promptbox}{User Prompt}
You are a person with the following characteristics:
{persona}

Topic: "{topic}"

Here are 100 statements from people with diverse perspectives on this topic:

{statements_list}

Write a NEW bridging statement expressing your views on this topic. Your statement should:
- Reflect your background, values, and life experiences
- Synthesize key themes you observed across the discussion
- Aim to find common ground or bridge different viewpoints
- Be 2-4 sentences long
- NOT write in first-person
- NOT explicitly reference your identity or demographics (avoid "As a [X]...")
- Be self-contained (do not reference "the statements above" or "other people")

Write only the statement:
\end{promptbox}

\paragraph{Without persona, with deliberation round}
\begin{promptbox}{System Prompt}
You are a helpful assistant that generates statements. Return only the statement text, no JSON or additional commentary. For each query, please generate a set of five possible responses, each within a separate <response> tag. Each <response> must include a <text> and a numeric <probability>.
Please sample at random from the tails of the distribution, such that the probability of each response is less than 0.10.
\end{promptbox}

\begin{promptbox}{User Prompt}
Topic: "{topic}"

Here are 100 statements from people with diverse perspectives on this topic:

{statements_list}

Write a NEW bridging statement on this topic. Your statement should:
- Synthesize key themes across the different viewpoints
- Aim to find common ground or bridge different viewpoints
- Be 2-4 sentences long
- Be self-contained (do not reference "the statements above" or "other people")
\end{promptbox}

\paragraph{Without persona, without deliberation round}
\leavevmode 
\begin{promptbox}{System Prompt}
You are a helpful assistant that generates statements. Return only the statement text, no JSON or additional commentary. For each query, please generate a set of five possible responses, each within a separate <response> tag. Each <response> must include a <text> and a numeric <probability>.
Please sample at random from the tails of the distribution, such that the probability of each response is less than 0.10.
\end{promptbox}

\begin{promptbox}{User Prompt}
Topic: "{topic}"

Write a bridging statement on this topic. Your statement should:
- Aim to find common ground or bridge different viewpoints
- Be 2-4 sentences long
\end{promptbox}

\subsubsection{Building Preference Profiles}
Using iterative ranking (\Cref{ssc:mitigate-preference-degeneracy}), the first few rounds ask the model to pick the top and bottom \(K\) statements, and the final round asks the model to rank the remaining statements.

\paragraph{Initial Rounds}
\begin{promptbox}{User Prompt}
Topic: "{topic}"

Here are {n} statements (identified by 4-letter codes):
{stmt_lines}

From these {n} statements, identify:
1. Your TOP {k} most preferred (in order, most preferred first)
2. Your BOTTOM {k} least preferred (in order, least preferred last)

IMPORTANT: Do NOT simply list codes in the order they appear above.
Your preferences should reflect your persona's values and background.

Return JSON: {"top_{k}": ["code1", "code2", ...], "bottom_{k}": ["code1", "code2", ...]}
\end{promptbox}

Where \verb|stmt_lines| is formatted as:
\begin{promptbox}{User Prompt}
XXXX: "Statement text here..."
YYYY: "Another statement..."
\end{promptbox}

\paragraph{Final Round}
\begin{promptbox}{User Prompt}
Topic: "{topic}"

Here are {n} statements (identified by 4-letter codes):
{stmt_lines}

Rank ALL of these statements from most to least preferred.

IMPORTANT: Do NOT simply list codes in the order they appear above.
Your preferences should reflect your persona's values and background.

Return JSON: {"ranking": ["most_preferred", "second", ..., "least_preferred"]}
\end{promptbox}

\subsubsection{Winner Selection}

\paragraph{GPT-Select: Select from P Alternatives}
\begin{promptbox}{System Prompt}
You are a helpful assistant that selects consensus statements. Return ONLY valid JSON.
\end{promptbox}

\begin{promptbox}{User Prompt}
Here are {n} statements from a discussion:

{statements_text}

Which statement would be the best choice as a consensus/bridging statement?
Consider which one best represents a reasonable middle ground that could satisfy diverse perspectives.

Return your choice as JSON: {"selected_statement_index": <index>}
Where the value is the index (0-{n-1}) of the statement you select.
\end{promptbox}

Where \verb|statements_text| is formatted as:
\begin{promptbox}{Format}
Statement 0: Statement text here...

Statement 1: Another statement...
\end{promptbox}

\paragraph{GPT-Select+Rank: Select with Rankings}
\begin{promptbox}{System Prompt}
You are a helpful assistant that selects consensus statements. Return ONLY valid JSON.
\end{promptbox}

\begin{promptbox}{User Prompt}
Here are {n} statements from a discussion:

{statements_text}

Here are preference rankings from {n_voters} voters:

{rankings_text}

Based on both the statements and the preference rankings, which statement would be the best choice as a consensus/bridging statement?

Return your choice as JSON: {"selected_statement_index": <index>}
Where the value is the index (0-{n-1}) of the statement you select.
\end{promptbox}

Where \verb|rankings_text| is formatted as ($K=20$ sampled voters, full rankings over $P=20$ alternatives):
\begin{promptbox}{Format}
Voter 1: 5 > 3 > 1 > 8 > 2 > 0 > 7 > 6 > 4 > 9 > 10 > 11 > 12 > 13 > 14 > 15 > 16 > 17 > 18 > 19
Voter 2: 3 > 5 > 8 > 1 > 2 > 0 > 7 > 6 > 4 > 9 > 10 > 11 > 12 > 13 > 14 > 15 > 16 > 17 > 18 > 19
...
Voter 20: 8 > 3 > 5 > 1 > 2 > 0 > 7 > 6 > 4 > 9 > 10 > 11 > 12 > 13 > 14 > 15 > 16 > 17 > 18 > 19
\end{promptbox}

\paragraph{GPT-Select+Pers: Select with Personas}
\begin{promptbox}{System Prompt}
You are a helpful assistant that selects consensus statements. Return ONLY valid JSON.
\end{promptbox}

\begin{promptbox}{User Prompt}
Here are {n} statements from a discussion:

{statements_text}

Here are the {n_voters} voters who will be voting on these statements:

{personas_text}

Based on both the statements and the voter personas, which statement would be the best choice as a consensus/bridging statement?

Return your choice as JSON: {"selected_statement_index": <index>}
Where the value is the index (0-{n-1}) of the statement you select.
\end{promptbox}

Where \verb|personas_text| is formatted as ($K=20$ sampled voters, filtered to 7 key fields to save token costs):
\begin{promptbox}{Format}
Voter 1: age: 53
sex: Male
race: White alone
education: Master's degree
occupation: MGR-Education And Childcare Administrators
political views: Liberal
religion: Protestant

Voter 2: age: 54
sex: Male
race: White alone
education: Master's degree
occupation: ENT-News Analysts, Reporters, And Journalists
political views: Democrat
religion: Catholic

...

Voter 20: ...
\end{promptbox}

\paragraph{GPT-Full: Select from All 100}
\begin{promptbox}{System Prompt}
You are a helpful assistant that selects consensus statements. Return ONLY valid JSON.
\end{promptbox}

\begin{promptbox}{User Prompt}
Topic: {topic}

A group of participants submitted the following {n_all} statements on this topic:

{all_text}

Select the statement that would best serve as a consensus or bridging position - one that:
- Engages substantively with the topic
- Could be acceptable to participants with diverse viewpoints
- Avoids extreme or polarizing framing

Return your choice as JSON: {"selected_statement_index": <index>}
Where the value is the index (0-{n_all-1}) of the statement you select.
\end{promptbox}

The GPT-Full+Rank and GPT-Full+Pers variants add \verb|rankings_text| and \verb|personas_text| respectively, using the same format as GPT-Select.

\paragraph{GPT-Synthesize: Generate New Statement}
\begin{promptbox}{System Prompt}
You are a helpful assistant that generates consensus statements. Return ONLY valid JSON.
\end{promptbox}

\begin{promptbox}{User Prompt}
Topic: {topic}

Here are some existing statements from a discussion:

{statements_text}

Generate a NEW statement that could serve as a better consensus/bridging statement.
The statement should:
- Represent a reasonable middle ground that could satisfy diverse perspectives
- Be different from the existing statements but address the same topic
- Be clear and substantive (2-4 sentences)

Return your new statement as JSON: {"new_statement": "<your statement>"}
\end{promptbox}

The GPT-Synthesize+Rank and GPT-Synthesize+Pers variants add \verb|rankings_text| and \verb|personas_text| respectively as in GPT-Select.

\paragraph{GPT-Blind: Blind Generation}
\begin{promptbox}{System Prompt}
You are a helpful assistant that generates bridging statements. Return ONLY valid JSON.
\end{promptbox}

\begin{promptbox}{User Prompt}
Given the topic: "{topic}"

Generate a bridging statement that could serve as a consensus position on this topic.
The statement should:
- Represent a reasonable middle ground that could satisfy diverse perspectives
- Acknowledge different viewpoints while finding common ground
- Be clear and substantive (2-4 sentences)

Return your statement as JSON: {"bridging_statement": "<your statement>"}
\end{promptbox}

\end{document}